\documentclass[11pt]{article}

\usepackage{fullpage}

\usepackage{dsfont}
\usepackage{xspace}
\usepackage{tikz}
\usetikzlibrary{patterns}
\usepackage{latexsym}
\usepackage{natbib}
\usepackage{amsmath}
\usepackage{amsthm}
\usepackage{amssymb}
\usepackage{amsfonts}
\usepackage{mathrsfs}
\usepackage{aliascnt}
\usepackage{pstricks}
\usepackage{pst-all}
\usepackage{pstricks-add}
\usepackage{pst-plot}
\usepackage{graphicx}
\usepackage{subfig}
\usepackage{float}
\usepackage[multiple]{footmisc}

\usepackage{enumerate}

\usepackage[pdfpagelabels,pdfpagemode=None]{hyperref}
\usepackage{cleveref}
\usepackage{algorithmic}
\usepackage{algorithm}

\newcommand{\STOC}[1]{}

\newcommand{\JC}[1]{{\linespread{1}\marginpar{\red Jason: #1}}}

\newcommand{\argmax}{\operatorname{arg\,max}}

\newtheorem{theorem}{Theorem}%

\newaliascnt{lemma}{theorem}
\newtheorem{lemma}[lemma]{Lemma}%
\aliascntresetthe{lemma}

\newaliascnt{claim}{theorem}
\aliascntresetthe{claim}

\newaliascnt{corollary}{theorem}
\aliascntresetthe{corollary}

\newaliascnt{proposition}{theorem}
\newtheorem{proposition}[proposition]{Proposition}%
\aliascntresetthe{proposition}

\newaliascnt{remark}{theorem}

\aliascntresetthe{remark}

\newaliascnt{algo}{procedure}
\aliascntresetthe{algo}

\theoremstyle{definition}
\newtheorem{definition}{Definition}

\newtheorem{example}{Example}

\makeatletter
\newtheorem*{rep@theorem}{\rep@title}
\newcommand{\newreptheorem}[2]{%
\newenvironment{rep#1}[1]{%
 \def\rep@title{#2 \ref{##1}}%
 \begin{rep@theorem}}%
 {\end{rep@theorem}}}
\makeatother
\newreptheorem{theorem}{Theorem}
\newreptheorem{lemma}{Lemma}

\newcommand{\aref}[1]{\hyperref[#1]{Appendix~\ref{#1}}}

\newcommand{\AutoAdjust}[3]{\mathchoice{ \left #1 #2  \right #3}{#1 #2 #3}{#1 #2 #3}{#1 #2 #3} }
\newcommand{\Xcomment}[1]{{}}

\newcommand{\InBrackets}[1]{\AutoAdjust{[}{#1}{]}}
\newcommand{\Ex}[2][]{\operatorname{\mathbf E}_{#1}\InBrackets{#2}}
\newcommand{\Prx}[2][]{\operatorname{\mathbf{Pr}}_{#1}\InBrackets{#2}}

\def\expect{\Ex}

\newcommand{\dd}{\mathrm{d}}  
\newcommand{\given}{\;\mid\;}

\newcommand{\bfzero}{{\boldsymbol{0}}}
\newcommand{\typeprelim}{{\boldsymbol{t}}}
\newcommand{\quantile}{{\boldsymbol{q}}}
\newcommand{\alloc}{\boldsymbol{x}}
\newcommand{\sagentmech}{\boldsymbol{x}}
\newcommand{\sagentmechelement}{{x}}

\newcommand{\sagentmechsdim}{x}

\newcommand{\reals}{\mathbb R}

\newcommand{\sagentmechcor}[1]{\boldsymbol{x}^{#1}}
\newcommand{\sagentmechcorelement}[1]{x^{#1}}
\newcommand{\denscor}[1]{\dens^{#1}}
\newcommand{\alphavcor}[1]{\alphav^{#1}}

\newcommand{\partiald}[2]{\partial_{#1}#2}

\newcommand{\sagentutil}{u}

\newcommand{\virtualv}{\boldsymbol{\alpha}}

\newcommand{\alphav}{\boldsymbol{\phi}}

\newcommand{\vectorfield}{\boldsymbol\alpha}
\newcommand{\vectorfieldcomponent}{\alpha}
\newcommand{\ironedvirtualv}{\bar{\boldsymbol{\phi}}}

\newcommand{\ironed}{\hat{\boldsymbol{\phi}}}

\newcommand{\sdamortil}{\alpha}
\newcommand{\sdamorev}{\phi}

\newcommand{\sdpamorev}{\sdamorev_{\max}}
\newcommand{\sdpdens}{\dens_{\max}}
\newcommand{\sdpdist}{\dist_{\max}}

\newcommand{\sumamorev}{\sdamorev_{\textup{sum}}}

\newcommand{\cumdiffelement}{\Gamma}
\newcommand{\cumdiff}{\mathbf{\cumdiffelement}}

\newcommand{\sumdens}{\dens_{\textup{sum}}}
\newcommand{\sumdist}{\dist_{\textup{sum}}}

\newcommand{\virt}{\ironedvirtualv}
\newcommand{\virtelement}{\bar{\phi}}
\newcommand{\amortil}{\virtualv}
\newcommand{\amortilelement}{\alpha}
\newcommand{\amorev}{\alphav}
\newcommand{\amorevelement}{\phi}
\newcommand{\dens}{f}
\newcommand{\typespace}{T}
\newcommand{\allocspace}{X}

\newcommand{\type}{\typeprelim}
\newcommand{\alloci}{x_i}
\newcommand{\util}{\sagentutil}
\newcommand{\typeboundary}{\partial\typespace}
\newcommand{\normal}{\boldsymbol{\eta}}
\newcommand{\sdnormal}{{\eta}}

\newcommand{\cost}{c}

\newcommand{\corrcurve}{C_{\textup{cor}}}
\newcommand{\thetamonotone}{ratio-monotone\xspace}
\newcommand{\thetamonotonicity}{ratio-monotonicity\xspace}

\newcommand{\curve}{C}
\newcommand{\completecorr}{perfectly correlated\xspace}
\newcommand{\quasisymmetric}{max-symmetric\xspace}
\newcommand{\quasisymmetry}{max-symmetry\xspace}

\newcommand{\boundary}{boundary inflow\xspace}

\newcommand{\sortedtype}{{\boldsymbol{v}}}

\newcommand{\constrained}{\hat}
\newcommand{\optconstrained}{\composed{\optimized}{\constrained}}
\newcommand{\optimized}{\starred}
\newcommand{\differentiated}[1]{#1'}
\newcommand{\fortype}{\tilde}

\newcommand{\starred}[1]{#1^\star}
\newcommand{\noaccents}[1]{#1}
\newcommand{\composed}[3]{#1{#2{#3}}}

\newcommand{\forexquant}[1]{#1^{\exquant}}

\newcommand{\newagentvar}[3][\noaccents]{%
\expandafter\newcommand\expandafter{\csname #2\endcsname}{#1{#3}}%
\expandafter\newcommand\expandafter{\csname #2s\endcsname}{#1{\boldsymbol{#3}}}%
\expandafter\newcommand\expandafter{\csname #2smi\endcsname}[1][i]{#1{\boldsymbol{#3}}_{-##1}}%
\expandafter\newcommand\expandafter{\csname #2i\endcsname}[1][i]{#1{#3}_{##1}}%
\expandafter\newcommand\expandafter{\csname #2ith\endcsname}[1][i]{#1{#3}_{(##1)}}%
}

\newagentvar[\tilde]{talloc}{\alloc}

\newagentvar{quant}{q}
\newagentvar[\constrained]{exquant}{\quant}
\newagentvar[\constrained]{critquant}{\quant}  
\newagentvar[\optconstrained]{monoq}{\quant}

\newagentvar{qprice}{\price}
\newagentvar{qrev}{R}
\newagentvar[\ironed]{iqrev}{\qrev}

\newagentvar{qalloc}{\alloc}
\newagentvar{cumalloc}{X}
\newagentvar[\constrained]{calloc}{\qalloc}
\newagentvar[\composed{\forexquant}{\constrained}]{callocstep}{\qalloc}
\newagentvar[\forexquant]{qallocstep}{\qalloc}
\newagentvar[\constrained]{cumcalloc}{\cumalloc}
\newagentvar[\ironed]{ialloc}{\qalloc}
\newagentvar[\ironed]{icumalloc}{\cumalloc}

\newagentvar{outcome}{w}
\newagentvar[\fortype]{toutcome}{\outcome}
\newagentvar[\composed{\forexquant}{\fortype}]{toutcomestep}{\outcome}

\newagentvar{rev}{R}
\newagentvar[\differentiated]{marg}{\rev}
\newagentvar{rawrev}{P}
\newagentvar[\differentiated]{rawmarg}{\rawrev}

\newagentvar[\tilde]{pseudorev}{\rev}
\newagentvar[\tilde]{pseudorawrev}{\rawrev}

\newagentvar{val}{v}

\newagentvar{Val}{V}

\newagentvar{price}{p}

\newagentvar[\constrained]{critval}{\val}
\newagentvar[\constrained]{reserve}{\val} 
\newagentvar[\optconstrained]{monop}{v}

\newagentvar{strat}{s}
\newagentvar{bid}{b}

\newagentvar{dist}{F}

\newagentvar[\ironed]{iprice}{\price}  
\newagentvar[\ironed]{ival}{\val}  

\newagentvar{ints}{{\cal I}}

\let\paragraphwithoutperiod\paragraph
\renewcommand{\paragraph}[1]{\paragraphwithoutperiod{#1.}}

\title{Multi-dimensional Virtual Values and \protect\\Second-degree Price Discrimination}

\author{Nima Haghpanah \\
MIT \\
EECS and Sloan School of Management \\
{\tt nima@csail.mit.edu} 
\and Jason Hartline \\
Northwestern University \\
EECS Department \\
{\tt hartline@northwestern.edu} 
}

\begin{document}

\begin{titlepage}
\maketitle
\begin{abstract}

We consider a multi-dimensional screening problem of selling a product with multiple quality levels and design virtual value functions to derive conditions that imply optimality of only selling highest quality. A challenge of designing virtual values for multi-dimensional agents is that a mechanism that pointwise optimizes virtual values resulting from a general application of integration by parts is not incentive compatible, and no general methodology is known for selecting the right paths for integration by parts.  We resolve this issue by first uniquely solving for paths that satisfy certain necessary conditions that the pointwise optimality of the mechanism imposes on virtual values, and then identifying distributions that ensure the resulting virtual surplus is indeed pointwise optimized by the mechanism. Our method of solving for virtual values is general, and as a second application we use it to derive conditions of optimality for selling only the grand bundle of items to an agent with additive preferences.

\end{abstract}

\thispagestyle{empty}
\end{titlepage}

\newpage

\section{Introduction}
\label{s:intro}

A monopolist seller can extract more of the surplus from consumers
with heterogeneous tastes through second-degree price discrimination.
While the optimal mechanism for a non-differentiated product is a
posted pricing, optimal mechanisms for a differentiated product can be
complex and even generally require the pricing of lotteries over the
variants of the product.  This paper gives sufficient conditions under
which the simple pricing of a non-differentiated product is optimal
even when product differentiation is possible.  These conditions allow
multi-dimensional tastes to be projected to a single dimension
where the pricing problem is easily solved by the classic theory.  The
identified conditions are natural and far more comprehensive than the
previous known conditions.  

The main technical contribution of the paper, from which these
sufficient conditions are identified, is a method for proving the
optimality of a family of mechanisms for agents with multi-dimensional
preferences.  This method extends the single-dimensional theory of
virtual values of \citet{M81} to multi-dimensional preferences.  The
main challenge of multi-dimensional mechanism design is that the paths
(in the agent's type space) on which the incentive constraints bind is
a variable; thus a straightforward attempt to generalize
single-dimensional virtual values to multi-dimensional agents is under
constrained.  To resolve this issue we introduce an additional
constraint on the virtual value functions that is imposed by the
optimality of mechanism in the family if point-wise optimization of
virtual values is indeed to result in a such a mechanism.  This
constraint pins down a degree of freedom in the derivation of virtual
value functions.  The family of mechanisms is optimal if there exists
virtual values that satisfy the additional as well as the standard
constraints on virtual values.  Importantly, this framework leaves the
paths on which the incentive constraints bind as a variable and solves
for them.

Consider a monopolist who can sell a high-quality or low-quality
product.  The values of a consumer for these differentiated products
can be seen as a point in the plane.  It will be convenient to write
the consumer's value for these two versions of the product as a base
value for the high-quality product and the same base value times a
discount factor for the low-quality product.  It is a standard result
of \citet{Sto79} and \citet{RZ-83} (and of \citealp{M81}, more
generally) that when the base value is private but the discount factor
is public, i.e., the values of the agent for the two qualities of
products are distributed on a line through the origin, then selling
only the high-quality good is optimal (and it is done by a posted
price).  The analysis of \citet{A96}, applied to this setting,
generalizes this result to the case where the base value and discount
factor are independently distributed but both private to the agent.
His result follows from solving the problem on every line from the
origin, as if the discount factor was public, and observing that these
solutions are consistent, i.e., they do not depend on the discount
factor, and therefore the same mechanism is optimal even when the
discount factor is private.

Our sufficient conditions generalize these results further to
distributions where the base value and discount factor are positively
correlated.\footnote{In this paragraph we assume that the marginal
  distribution of the base value is regular, i.e., Myerson's virtual
  value is monotone; and positive correlation is defined by
  first-order stochastic dominance.  Generalizations are given later
  in the paper.}  Notice that allowing arbitrary correlations between
base value and discount factor is completely general as a
multi-dimensional screening problem for a high- and low-quality
product.  Further, the example of \citet{Tha04}, which we review in
detail subsequently, shows that the single-dimensional projection,
i.e., selling only the high-quality product is not generally optimal
with correlated base value and discount factor.  Consider the special
case where base value and discount factor are perfectly correlated,
i.e., the values for the differentiated products lie on a curve from
the origin.  In this case, the agent's type is actually
single-dimensional but her tastes are multi-dimensional.  We prove
that if the curve only crosses lines from the origin from below, i.e.,
the discount factor is monotonically non-decreasing in the base value,
then selling only the high-quality product is optimal.  On the other
hand, if the discount factor is not monotone in the base value then we
show that there exists a distribution for the base value for which it
is not optimal to sell only the high-quality product.  Perfect
correlation with a monotone discount factor is a special case of
positive correlation which we show remains a sufficient condition for
optimality of selling only the high-quality product.

From the analysis of the perfectly correlated case, we see that the
analyses of \citet{A96} where the discount factor is
independent of the base value, and \citet{Sto79} and \citet{RZ-83} where the discount factor is known, are at the boundary between optimality
and non-optimality of selling only the high-quality product.  Thus,
these results are brittle with respect to perturbations in the model.
Our result shows that pricing only the high-quality product remains
optimal for any positive correlation; the more positively correlated
the model is the more robust the result is to perturbations of the
model.

Our characterization of positive correlation of the base value and
discount factor as sufficient for the optimality of selling only the
high-quality product is intuitive.  Price discrimination can be
effective when high-valued consumers are more sensitive to quality
than low-valued consumers.  These high-valued consumers would then
prefer to pay a higher price for the high-quality product than to
obtain the low-quality product at a lower price.  Positive correlation
between the base value and discount factor eliminated this
possibility.  It implies that high-valued agents are less sensitive to
quality than low-valued agents.  


As a qualitative conclusion from this work, optimal second-degree
price discrimination, which is complex in general, cannot improve a
monopolists revenue over a non-differentiated product unless
higher-valued types are more sensitive (with respect to the ratio of
their values for high- and low-quality products) to product
differentiation than lower-valued types.  This simplification, for
consumers that exhibit positive correlation, generalizes from monopoly
pricing to general mechanism design.  For example, a (monopolist)
auctioneer on eBay has no advantage of discriminating based on
expedited or standard delivery method if high-valued bidders discount
delayed delivery less than low valued bidders. 

The above characterizations show that the multi-dimensional pricing
problem reduces to a single-dimensional projection where the agent's
type is, with respect to the examples above, her base value.  Our
proof method instantiated for this problem is the following.  We need
to show the existence of a virtual value function for which (a)
point-wise optimization of virtual surplus gives a mechanism that
posts a price for the high-quality product and (b) expected virtual
surplus equals expected revenue when the agent's type is drawn from
the distribution.  If the single-dimensional projection is optimal and
(b) holds then it must be that the virtual value of the high-quality
product is equal to the single-dimensional virtual value according to
the marginal distribution of agent's value for the high-quality
product.  This pins down a degree of freedom in problem of identifying
a virtual value function (which is generally given by integration by
parts on the paths in type space, e.g., \citealp{RC98}); the virtual
value for the low-quality product can then be solved for from the
high-quality virtual value and a differential equation that relates
them.  It then suffices to check that (a) holds, which in this case
requires that, at any pair of values for the high and low qualities,
(a.1) if the virtual value for the high-quality product is positive
then it is at least the virtual value of the low-quality product and
(a.2) if it is negative then they are both negative.  Analysis of the
constraints imposed by (a.1) and (a.2) then gives sufficient
conditions on the distribution on types for optimality of the
single-dimensional projection.

Our result above applies generally to a risk-neutral agent with quasi-linear utility over multiple outcomes, and identifies conditions for optimality of a mechanism that simply posts a uniform price for all outcomes (i.e., the only non-trivial outcome assigned to each type is its favorite outcome).  Applied to setting where the consumer can buy multiple items and outcomes correspond to bundles of items, this result indirectly gives conditions for optimality of posting a price for the grand bundle of items.  If a uniform price is posted for all bundles, the consumer will only buy the grand bundle, or nothing (assuming free disposal). 

The special case of this bundle pricing problem where the consumer's
values are additive across the items has received considerable
attention in the literature \citep{AdY76, HaN12, DDT14-2} and our
framework for proving optimality of single-dimensional projections can
be applied to it directly.  For this application, we employ a more
powerful method of virtual values which is analogous to the ironing
approach of \citet{M81}.  We show that, for selling two items to a
consumer with additive value, grand-bundle pricing is optimal when
higher value for the grand bundle is negatively correlated with the
ratio of values for the two items, i.e., when higher valued consumers
have more heterogeneity in their tastes.  This result formalizes a
connection that goes back to \citet{AdY76}.  This second
application of our framework for proving the optimality of simple
mechanisms further demonstrates its general applicability.

\subsection{Related Work}

The starting point of work in multi-dimensional optimal mechanism
design is the observation that an agent's utility must be a convex
function of his private type, and that its gradient is equal to the allocation \citep[e.g.,][cf.\@ the envelope
  theorem]{Roc85}.  The second step is in writing revenue as the
difference between the surplus of the mechanism and the agent's
utility \citep[e.g.,][]{MM88,A96}.  The surplus can be expressed in terms
of the gradient of the utility.  The third step is in rewriting the
objective in terms of either the utility
\citep[e.g.,][]{MM88,MV06,HaN12,DDT13,WT-13,GK14} or in terms of the
gradient of the utility (e.g., \citealp{A96}; \citealp{AFHH13}; and
this paper).  This manipulation follows from an integration by parts.
The first category of papers (rewriting objective in terms of utility)
performs the integration by parts independently in each dimension, and
the second category (rewriting objective in terms of gradient of
utility, except for ours) does the integration along rays from the
origin.  In our approach, in contrast, integration by
parts is performed in general and is dependent on the distribution and
the form of the mechanism we wish to show is optimal.

Closest to our work are \cite{Wil93}, \citet{A96}, and \citet{AFHH13}
which use integration by parts along paths that connect types with
straight lines to the zero type (which has value zero for any outcome)
to define virtual values. \cite{Wil93} and \citet{A96} gave closed
form solutions for multi-dimensional screening problems.  Their
results are for nonlinear problems that are different from our model.
\citet{AFHH13} used integration by parts to get closed form solutions
with independent and uniformly distributed values; our results
generalize this one. Importantly, the paths for integration by parts
in all these works is fixed a priori.  In contrast, the choice of
paths in our setting varies based on the distribution. \citet{RC98}
showed that the general application of integration by parts (with
parameterized choice of paths) characterizes the solutions of the
relaxed problem where all but local incentive constraints are removed.
However, the characterization is implicit and includes the choice of
paths as parameters.  They use the characterization to show that since
bunching can happen, the solution to the relaxed problem is
generically not incentive compatible.\footnote{Bunching refers to the
  case where different types are assigned the same allocation.}
Importantly, the observation is based on the placement of the outside
option, in the form of a price for a certain allocation, that is the
zero allocation in our setting.  Compared to the above papers, our
work is the first to use the variability of paths to derive explicit
conditions of optimality (see \citealp{RoS03}, for an accessible
survey).

There has been work looking at properties of single-agent mechanism
design problems that are sufficient for optimal mechanisms to make
only limited use of randomization.  For context, the optimal
single-item mechanism is always deterministic
\citep[e.g.,][]{M81,RZ-83}, while the optimal multi-item mechanism is
sometimes randomized \citep[e.g.,][]{Tha04,Pyc06}.  For agents with
additive preferences across multiple items, \citet{MM88},
\citet{MV06}, and \citet{GK14} find sufficient conditions under which
deterministic mechanisms, i.e., bundle pricings, are optimal.
\citet{Pav11} considers more general preferences and a more general
condition; for unit-demand preferences, this condition implies that in
the optimal mechanism an agent deterministically receives an item or
not, though the item received may be randomized. Our approach is
different from these works on multi-dimensional mechanism design in
that it uses properties of a pre-specified family of mechanisms to pin
down multi-dimensional virtual values that prove that mechanisms from
the family are optimal.

A number of papers consider the question of finding closed forms for
the optimal mechanism for an agent with additive preferences and
independent values across the items.  One such closed form is
grand-bundle pricing.  For the two item case, \citet{HaN12} give
sufficient conditions for the optimality of grand-bundle pricing;
these conditions are further generalized by \citet{WT-13}.  Their
results are not directly comparable to ours as our results apply to
correlated distributions.  \citet{DDT14-2} and \citet{GK14} give
frameworks for proving optimality of multi-dimensional mechanisms, and
find the optimal mechanism when values are i.i.d.\@ from the uniform
distribution (with up to six items).  \citet{DDT14-2} establish a
strong duality theorem between the optimal mechanism design problem
with additive preferences and an optimal transportation problem
between measures (similar to the characterization of
\citealp{RC98}. Using this duality they show that every optimal
mechanism has a certificate of optimality in the form of
transformation maps between measures. They use this result to show
that when values for $m\geq 2$ items are independently and uniformly
distributed on $[c,c+1]$ for sufficiently large $c$, the grand
bundling mechanism is optimal, extending a result of \citet{Pav11} for
$m=2$ items. In comparison, a simple corollary of our theorem states
that grand bundling is optimal for uniform draws from $[a,b]$
truncated such that the sum of the values is at most $a+b$, for
\emph{any} $a\leq b$.

\section{Preliminaries}
\label{s:prelim}

We consider a single-agent mechanism design problem with allocation space $X\subseteq [0,1]^m$, and a bounded connected type space $T \subset \reals^m$ with Lipschitz continuous boundary, for a finite $m$. The utility of the agent with type
 $\typeprelim \in \typespace$ for allocation $\alloc \in \allocspace$ and payment $p\in
 \reals$ is $\typeprelim \cdot \alloc - \price$.\footnote{Throughout the paper we maintain the convention
   of denoting a vector $\bold{v}$ by a bold symbol and each of its
   components $v_i$ by a non-bold symbol.} Our main results are for the following outcome spaces $X$.
 \begin{itemize}
 \item \emph{The multi-outcome setting} (\autoref{s:uniform}): We
   assume $\allocspace = \{\alloc \in [0,1]^m \mid \sum_i \alloci \leq
   1\}$. Here $m$ is the number of outcomes, and an allocation is a
   distribution over outcomes ($1-\sum_i \alloci$ is the probability of
   selecting a null outcome for which the agent has zero value).  For
   example, $m$ may be the number of possible configurations, e.g.,
   quality or delivery method, of a single item to be sold.  As
   another example, $m=2^k$ may be the number of possible bundles of
   $k$ items to be allocated.
 \item \emph{The multi-product setting with additive preferences}
   (\autoref{a:add bundle}): We assume $\allocspace = [0,1]^m$.  Here
   $m$ is the number of items, and an allocation specifies the
   probability $\alloci$ of receiving each item.\footnote{This setting is
     a special case of the multi-outcome setting with $2^m$ outcomes.
     The additivity structure allows us to focus on the lower
     dimensional space of items instead of outcomes.}
 \end{itemize}
The {\em cost} to the seller for producing outcome $\alloc$ is denoted
$\cost(\alloc)$ and the sellers {\em profit} for $(\alloc,\price)$ is
$\price - \cost(\alloc)$.

We use the revelation principle and focus on direct mechanisms. A single-agent mechanism is a pair of functions, the allocation
function $\sagentmech: T \rightarrow \allocspace$ and the payment
function $p: T \rightarrow \mathbb R$. A mechanism is \emph{incentive compatible} (IC) if no type of the agent increases his utility by misreporting, 
\begin{align*}
\typeprelim \cdot \sagentmech(\typeprelim) - p(\typeprelim) &\geq \typeprelim \cdot \sagentmech(\hat{\typeprelim}) - p(\hat{\typeprelim}), &\qquad \forall \typeprelim, \hat{\typeprelim} \in T.
\intertext{A mechanism is
\emph{individually rational} (IR) if the utility of every type of the agent
is at least zero,}
\typeprelim \cdot \sagentmech(\typeprelim) - p(\typeprelim) &\geq 0, &\qquad \forall \typeprelim \in T.
\end{align*}

A single agent mechanism $(\sagentmech,p)$ defines a utility function $\sagentutil(\typeprelim) = \typeprelim \cdot \sagentmech(\typeprelim) - p(\typeprelim)$. The following lemma connects the utility function of an IC mechanism with its allocation function.

\begin{lemma}[\citealp{Roc85}]
Function $u$ is the utility function of an agent in an individually-rational incentive-compatible mechanism if and only if $u$ is convex, non-negative, and non-decreasing. The allocation is \em{$\sagentmech(\typeprelim) = \nabla \sagentutil(\typeprelim)$}, wherever the gradient $\nabla \sagentutil(\typeprelim)$ is defined.\footnote{If $u$ is convex, $\nabla \sagentutil(\typeprelim)$ is defined almost everywhere, and the mechanism corresponding to $u$ is essentially unique.}
\label{manelivincent}
\label{l:gradutil=alloc}
\end{lemma}

Notice that the payment function can be defined using the utility
function and the allocation function as $p(\typeprelim) = \typeprelim
\cdot \sagentmech(\typeprelim) - \sagentutil(\typeprelim)$.  Applying
the above lemma, we can write payment to be $p(\typeprelim) =
\typeprelim \cdot \nabla \sagentutil(\typeprelim) -
\sagentutil(\typeprelim)$.    In the profit maximization problem the objective is to maximize the expected revenue minus cost, when the types are drawn at random from a distribution over $\typespace$ with density $\dens>0$.  Using \autoref{l:gradutil=alloc}, the problem can
be written as the following mathematical program.
\begin{eqnarray}
\max_{\sagentmech, \sagentutil} && \int_\typeprelim \big[\typeprelim \cdot \sagentmech(\typeprelim) - \sagentutil(\typeprelim) - \cost(\sagentmech(\typeprelim))\big] f(\typeprelim)\; \dd \typeprelim \label{singleagentformulation} \\
&& \sagentutil \text{ is convex}; \sagentutil \geq 0 \nonumber \\
&& \nabla \sagentutil = \sagentmech \in X. \nonumber
\end{eqnarray}

The primary task of this paper is to identify condition that imply the
optimality of \emph{single-dimensional projection} mechanisms.  In a
single-dimensional projection mechanism the preferences can be
summarized by a mapping of the multi-dimensional type $\type$ into a
single-dimensional value.  In particular, in the multi-outcome setting
(\autoref{s:uniform}) we will study the optimality of the class of
\emph{favorite-outcome} projection mechanisms where $\alloci(\type)>0$
only if $i$ is the favorite outcome, $i = \arg \max_j t_j$.  For such
a mechanism, the only relevant information a type contains is the
value for the favorite outcome. In the multi-product setting with
additive preferences (\autoref{a:add bundle}) we study
\emph{sum-of-values} projection mechanisms where $\alloci(\type) =
x_j(\type)$ for all $i$ and $j$, and the value for the grand bundle
$\sum_i t_i$ summarizes the preferences.  The optimization over these
classes can be done using standard methods from single-dimensional
analysis \citep{M81,RZ-83}, where we know the optimal mechanism is
non-stochastic.  The optimal favorite-outcome projection mechanism is
a {\em uniform pricing}, i.e., the same price is posted on all
non-trivial outcomes; the optimal multi-product sum-of-values
projection mechanism is a {\em grand bundle pricing}, i.e., a price is
posted for the grand bundle only.  This paper develops a theory for
proving that these single-dimensional projections are optimal among
all multi-dimensional mechanisms.

The seller's cost $\cost(\alloc)$ for producting outcome $\alloc$ can
generally be internalized into the consumer's utility and thus
ignored.  In our analysis we will expose only a uniform {\em service
  cost}, i.e., $\cost$ when the agent is served any non-trivial
outcome (as discussed further in \autoref{sec:reverseengineering}).
For the multi-outcome setting this service cost can be written as
$\cost \sum_i \alloci$ and for the multi-product setting it can be
written as $\cost \max_i \alloci$.

\section{Amortizations and Virtual Values}
\label{s:amortized}


This section codifies the approach of incentive compatible mechanism
design via virtual values and extends it to agents with
multi-dimensional type spaces.  A standard approach to understanding
optimal mechanisms via multi-dimensional virtual values is to require
that virtual surplus equate to revenue for the optimal mechanism (see
survey by \citealp{RoS03}).  For this approach to be as directly
useful as it has been in single-dimensional settings, we depart from
this literature and impose two additional conditions.  We require
virtual surplus to relate to (in particular, as an upper
bound)\footnote{Relaxing from equality to an upper bound enables our
  analysis to (a) generalize to mechanisms without binding
  participation constraints and to (b) allow for a generalization of
  the ``ironing'' procedure of \citet{M81}.}  revenue for all
incentive compatible mechanisms, we call this condition {\em
  amortization};\footnote{This terminology comes from the design and
  analysis of algorithms in which an {\em amortized analysis} is one
  where the contributions of local decisions to a global objective are
  indirectly accounted for \citep[see the textbook of][]{BoE98}.  The
  correctness of such an indirect accounting is often proven via a
  {\em charging argument}.  Myerson's construction of virtual values
  for single-dimensional agents can be seen as making such a charging
  argument where a low type, if served, is charged for the loss in
  revenue from all higher types.} and that pointwise optimization of
virtual surplus without the incentive compatibility constraint gives
incentive compatibility for free.  The amortization conditions are
relatively easy to satisfy, essentially from integration by parts on paths
that cover type space; while the incentive compatibility condition is
not generally satisfied by an amortization.  An exception is the
single dimensional special case, with $m=1$ non-trivial outcome,
where the integration by parts is unique and often incentive
compatible.






\begin{definition}
\label{d:amortization} 
A vector field $\virt : \typespace \rightarrow \reals^m$ is 
an {\em amortization of revenue} if expected virtual surplus (without
costs)\footnote{Equivalently with costs, the
  same holds for expected profit, i.e., $\forall
  (\hat{\alloc},\hat{\price}),\ \expect{\virt(\type) \cdot
    \hat{\alloc}(\type) - \cost(\hat{\alloc}(\type))} \geq
  \expect{\hat{\price}(\type) - \cost(\hat{\alloc}(\type))}$.} is an upper bound on the expected revenue of all
individually-rational incentive-compatible mechanisms, i.e., $\forall
(\hat{\alloc},\hat{\price}),\ \expect{\virt(\type) \cdot
  \hat{\alloc}(\type)} \geq
\expect{\hat{\price}(\type)}$;  it is
%
{\em tight} for
  incentive-compatible mechanism $(\alloc,\price)$ if the inequality
  above is tight, i.e., $\expect{\virt(\type) \cdot \alloc(\type)} =
  \expect{\price(\type)}$.
\end{definition}

\begin{definition}
\label{d:md-vvf}
An amortization of revenue $\virt : \typespace \rightarrow \reals^m$ is a
{\em virtual value function} if a {\em pointwise virtual surplus
  maximizer} $\alloc$, i.e., $\alloc(\type) \in
\argmax_{\hat{\alloc}\in \allocspace} \hat{\alloc} \cdot \virt(\type)
- \cost(\hat{\alloc}), \forall \type
\in \typespace$,\footnote{Often this virtual surplus maximizer
  is unique up to measure zero events, when it is not then these
  conditions must hold for one of the virtual surplus maximizers and
  we refer to this one as {\em the} virtual surplus maximizer.} is incentive
compatible and tight for $\virt$, i.e., there exists a payment rule $\price$ such
that the mechanism $(\alloc,\price)$ is incentive compatible, individually rational,
and tight for $\virt$.
\end{definition}

\begin{proposition}
\label{prop:vvf=>IC+opt}
\label{unironedmeta}
\label{ironedmeta}
For any mechanism design problem that admits a virtual
value function, the virtual surplus maximizer is the optimal
mechanism.
\end{proposition}

\begin{proof}
Denote the virtual surplus maximizer of \autoref{d:md-vvf} by $(\alloc,\price)$ and any
alternative IC and IR mechanism by $(\hat{\alloc},\hat{\price})$; then,
\begin{align*}
\expect{\price(\type) - c(\alloc(\type))} 
  & =
\expect{\virt(\type) \cdot \alloc(\type) - c(\alloc(\type))} 
  \\& \geq 
\expect{\virt(\type) \cdot \hat{\alloc}(\type) - c(\hat{\alloc}(\type))} 
  \geq
\expect{\hat{\price}(\type) - c(\hat{\alloc}(\type))}.
\end{align*}
The expected profit of the mechanism is equal to
its expected virtual surplus (by tightness).  This expected virtual surplus is at least
the virtual surplus of any alternate mechanism (by pointwise
optimality).  The expected virtual surplus of the
alternative mechanism is an upper bound on its expected profit
(an amortization gives an upper bound on expected profit).
\end{proof}

\subsection{Canonical Amortizations}
There is a canonical family of amortizations given by writing expected
utility as an integral and integrating it by parts on paths that cover
type space.  Intuitively, this integration by parts attributes to each type
on a path the loss in revenue from all higher types on the path when
this type is served by the mechanism.  For single-dimensional agents the path is unique and,
thus, so is the canonical amortization \citep{M81}; for
multi-dimensional agents neither paths nor canonical amortizations are
unique.  The latter integration by parts on
paths can be expressed as a multi-dimensional integration by parts
with respect to a vector field $\amortil$ that satisfies two
properties (see \citealp{RC98}):
\begin{itemize}
\item \emph{divergence density equality}: $\nabla \cdot \amortil = -\dens$ for all types $\type \in \typespace$, and
\item \emph{\boundary}: $(\amortil \cdot \normal)(\type)\leq 0$ for
  all types $\type \in \typeboundary$ where $\typeboundary$ denotes the boundary of type space $\typespace$.
\end{itemize}
The divergence density equality condition requires $\amortil$ to
correspond to distributing the required density $\dens$ on paths.  In
the integration by parts on paths, intuitively, each path begins with
an inflow of probability mass, and it distributes this along the path
according to the density function to the end of the path.  Thus, the
direction of $\amortil(\type)$ is the direction of the path at type
$\type$ and its magnitude is the remaining probability mass to be
distributed on the path.  The initial inflow of probability mass at
the origin of the path should be set so that none is left when the
path terminates.  With this interpretation the \boundary condition is
satisfied: on boundary types that originate paths there is an inflow,
on boundary types $\type$ parallel to paths the dot product $(\amortil
\cdot \normal)(\type)$ is zero, and on boundary types that terminate
paths the magnitude is
zero.\footnote{\label{footnote:canonical-exceptions}The \boundary
  condition also allows inflow at the terminal types, the amortization
  from such an $\amortil$ will not generally be tight for non-trivial
  mechanisms.  Without loss we do not consider amortizations
  constructed from such $\amortil$ to be canonical here, or below in
  \autoref{def:amortiltoamorev}.}  See
\autoref{fig:intbypartsintuition}.  The following lemma recasts a
result of \citet{RC98} into our framework.
\begin{figure}
\centering
\begin{tikzpicture}[scale = 3.1, align=center]
    \filldraw[black!40] plot [smooth] coordinates {(0.4,.08) (0.9,.4)} -- (0.9,.7) -- plot [smooth] coordinates {(.2,.47) (.3,.4) (.38,.2) (.4,.08)};
    \draw[line width=1.5] plot [smooth] coordinates {(0,.1) (0.3,.025) (0.9,.4)};
    \draw[line width=1.5] (0.9,.4)--(0.9,.7) -- (0,.4) -- (0,.1);
    \draw[line width=2.5]  plot [smooth] coordinates {(0,.52) (.1, .5) (.2,.47) (.3,.4) (.38,.2) (.4,.08) (.4,0)};
    \draw (0,.52) node[right,xshift=15pt,yshift=10pt]{$\util >0$};
    \draw (0,.52) node[left,yshift=-8pt,xshift=1pt]{$\util=0$};

    \draw[line width=1]  plot [smooth] coordinates {(0,.3) (.2,.2) (.5,.35) (.9,.6)};
    \draw[->,line width=2] (0,.3) -- (.15,.21);
    \filldraw (0,.3) circle(.5pt);
    
    \draw[->,line width=2] (.5,.35) -- (.6,.41);
    \filldraw  (.5,.35) circle(.5pt) node[above,yshift=3pt]{$\amortil$};
    
     \filldraw  (.9,.6) circle(.5pt) node[right]{$\sdamortil_1=0$};

\end{tikzpicture}
\caption{\footnotesize $\amortil/\dens$ is the loss in revenue from all higher types on the path.}
\label{fig:intbypartsintuition}
\end{figure}
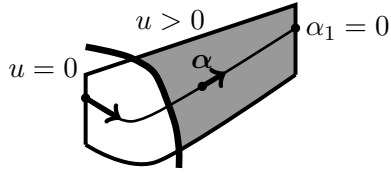

\begin{lemma} 
\label{lem:framework}
For a vector field $\amortil : \typespace \rightarrow \reals^m$
satisfying the divergence density equality and \boundary, the vector
field $\amorev(\type) = \type -
\amortil(\type)/\dens(\type)$ is an amortization of revenue; moreover,
it is tight for any incentive compatible mechanism for which the
participation constraint is binding for all boundary types with strict
inflow, i.e., $\util(\type) = 0$ for $\type \in \typeboundary$ with
$(\amortil \cdot \normal)(\type) < 0$.\footnote{\citet{RC98} prove
  this lemma by taking the first order conditions of
  program~\eqref{singleagentformulation} relaxing the constraint that
  utility is convex.  A result of such analysis is that $\amortil$ can
  be alternatively viewed as the Lagrangians of the local incentive
  compatibility constraints.  We will mainly focus on the
  interpretation of $\amortil$ as the direction of paths for
  integration by parts.}
\end{lemma}
\begin{proof} 
The following holds for any incentive compatible mechanism.
Integration by parts allows expected utility $\expect{\util(\type)}$
to be rewritten in terms of gradient $\nabla \util$ and vector
field $\amortil$ satisfying the divergence density
equality.\footnote{Integration by parts for functions $h: \reals^k
  \rightarrow \reals$ and $\vectorfield: \reals^k \rightarrow
  \reals^k$ over a set $\typespace$ with Lipschitz continuous boundary is as follows
\begin{align*}
 \int_{\typeprelim \in T} (\nabla h \cdot \vectorfield)(\typeprelim) \; \dd \typeprelim =- \int_{\typeprelim\in T} h(\typeprelim) (\nabla \cdot \vectorfield(\typeprelim)) \; \dd \typeprelim+  \int_{\typeprelim \in \partial T} h(\typeprelim)(\vectorfield \cdot \boldsymbol{\eta})(\typeprelim) \; \dd \typeprelim,
\end{align*}
\noindent where $\nabla \cdot \vectorfield(\type)$ is the divergence
of $\vectorfield$ and is defined as $\nabla \cdot \vectorfield =
\partiald{1}{\vectorfieldcomponent_1} + \ldots +
\partiald{k}{\vectorfieldcomponent_k}$, and $\normal(\type)$ is the normal to the boundary at $\type$.}
\begin{align}
\int_{\type \in \typespace} \nabla \util(\type)\cdot \amortil(\type) \, \dd\type
   &=  - \int_{\type\in \typespace} \util(\type)\, (\nabla \cdot \amortil(\type)) \, \dd \typeprelim
   +  \int_{\type \in \typeboundary} \util(\type) (\amortil \cdot \normal)(\type) \, \dd\type \nonumber \\
      & = \int_{\type\in \typespace} \util(\type)\, \dens(\type) \, \dd \typeprelim + \int_{\type \in \typeboundary} \util(\type) (\amortil \cdot \normal)(\type) \, \dd\type. \nonumber 
   \intertext{By \autoref{l:gradutil=alloc}, which implies that the allocation rule of the mechanism is the gradient of the utility, i.e., $\alloc(\type) = \nabla\util(\type)$, and the definition of expectation:}
   \expect{\tfrac{\amortil(\type)}{\dens(\type)}\cdot\alloc(\type)} &= \expect{\util(\type)} + \int_{\type \in \typeboundary} \util(\type) (\amortil \cdot \normal)(\type) \, \dd\type.
 \label{eq:amortil}
   \intertext{Individual rationality implies that $\util(\type) \geq 0$ for all $\type \in \typespace$; combined with the assumed \boundary condition, the last term on the right-hand side is non-positive.  Thus,}
   \expect{\tfrac{\amortil(\type)}{\dens(\type)}\cdot\alloc(\type)} &\leq \expect{\util(\type)}\!. 
 \nonumber
\\
\intertext{Revenue is surplus less utility; thus, $\amorev(\type) = \type - \amortil(\type)/\dens(\type)$ is an amortization of revenue, i.e.,}
      \expect{\amorev(\type) \cdot \alloc(\type)}
  &\geq \expect{\price(\type)}\!.
\nonumber
\end{align}   
Finally, notice that if the last term of the right-hand side of
equation~\eqref{eq:amortil} is zero, which holds for all mechanisms
for which the individual rationality constraint is binding for types $\type$ on the boundary at which the paths
specified by $\amortil$ originate, then the inequalities above are
equalities and the amortization is tight.
\end{proof}

\begin{definition} 
\label{def:amortiltoamorev}
A {\em canonical amortization of revenue} is $\amorev(\type) = \type - \amortil(\type)/\dens(\type)$ with $\amortil$ satisfying the divergence density inequality and \boundary.$^{\ref{footnote:canonical-exceptions}}$
\end{definition}



For a single-dimensional agent with value $\val$ in type space
$\typespace = [\underline{\val},\bar{\val}]$, the canonical
amortization of revenue that is tight for any non-trivial mechanism is unique and given by
$\sdamorev(\val) = \val
- \frac{1-\dist(\val)}{\dens(\val)}$.\footnote{Divergence density
  equality implies that $\sdamortil(\val) =
  \sdamortil(\underline{\val}) - \dist(\val)$. Tightness requires that
  $\sdamortil(\bar{\val})\util(\bar{\val}) =
  (\sdamortil(\underline{\val}) -1) \util(\bar{\val})= 0$. Since
  $\util(\bar{\val})>0$ for any non-trivial mechanism, we must have
  $\sdamortil(\underline{\val}) =1$ and thus $\sdamortil(\val) = 1 -
  \dist(\val)$.  Tightness also requires that
  $\sdamortil(\underline{\val})\util(\underline{\val})=0$, which is
  satisfied for any mechanism with binding participation constraint
  $\util(\underline{\val})=0$.}  When it is monotone, pointwise virtual surplus maximization is incentive compatible, and thus the canonical amortization $\amorev$ is a virtual value function.

\subsection{Reverse Engineering Virtual Value Functions}
\label{sec:reverseengineering}
Multi-dimensional amortizations of revenue, themselves, do not greatly
simplify the problem of identifying the optimal mechanism as they are
not unique and in general virtual surplus maximization for such an
amortization is not incentive compatible.  The main approach of this
paper is to consider a family of mechanisms and to add constraints
imposed by tightness and virtual surplus maximization of this family of mechanisms
to obtain a unique amortization.  First,
we will search for a single amortization that is tight for all
mechanisms in the family.  Second, we will consider virtual surplus
maximization with a class of cost functions and require that a mechanism in
the family be a virtual surplus maximizer for each cost (see \autoref{s:prelim}).  These two
constraints pin down a degree of freedom in an amortization of
revenue and allow us to solve for the amortization uniquely. The remaining task is to identify the sufficient
conditions on the distribution such that a
mechanism in the family is a virtual surplus maximizer.  
Subsequently in \autoref{s:uniform}, we will identify sufficient
conditions on the distribution of types for the family of uniform pricing
mechanisms to be optimal.


Our framework also allows for proving the optimality of mechanisms
when no canonical amortization of revenue is a virtual value function.
In the single-dimensional case, the ironing method of \citet{MuR78}
and \citet{M81}, can be employed to construct, from the canonical
amortization $\amorev$, another (non-canonical) amortization $\virt$
that is a virtual value function.  The multi-dimensional
generalization of ironing, termed sweeping by \citet{RC98}, can
similarly be applied to multi-dimensional amortizations of revenue.
The goal of sweeping is to reshuffle the amortized values in $\amorev$
to obtain $\virt$ that remains an amortization, but additionally its
virtual surplus maximizer is incentive compatible and tight.  Our
approach in this paper will be to prove a family of mechanisms is
optimal for any uniform service costs by invoking the following
proposition, which follows directly from the definition of
amortization (\autoref{d:md-vvf}).
\begin{proposition}
\label{prop:amorev-to-virt}
A vector field $\virt$ is an amortization of revenue if, for all
incentive compatible mechanisms $(\hat{\alloc},\hat{\price})$ and some
other amortization of revenue $\amorev$, it satisfies
$\expect{\virt(\type)\cdot\hat{\alloc}(\type)} \geq
\expect{\amorev(\type)\cdot\hat{\alloc}(\type)}$.
\end{proposition}
We adopt the sweeping approach in \autoref{a:add bundle} (and
\autoref{thm:unitdemandironing} which extends the main result of
\autoref{s:uniform}).  Just as there are many paths in
multi-dimensional settings, there are many possibilities for the
multi-dimensional sweeping of \citet{RC98}.  Our positive results
using this approach will be based on very simple single-dimensional
sweeping arguments.


\section{Optimality of Favorite-outcome Projection}
\label{s:uniform}
In this section we study conditions that imply a favorite-outcome projection mechanism is optimal in the multi-outcome setting (see \autoref{s:prelim}).  In that case, the problem collapses to a monopoly problem with a single parameter (the value for the favorite outcome), where we know from \citet{RZ-83} that the optimum mechanism is \emph{uniform pricing}: all nontrivial outcomes are deterministically and uniformly priced.  
 
As discussed in \autoref{sec:reverseengineering}, we use a class of cost functions to restrict the admissible amortizations. Throughout this section we assume uniform constant marginal costs, that is, $\cost(\alloc) = c \sum_i x_i$ for some constant service cost $c\geq0$.\footnote{Any instance with non-uniform marginal costs can be converted to an instance with zero cost by redefining value as value minus cost.}  For simplicity we focus on the case of two outcomes (extension in \autoref{s:extensions}).  To warm up, we will start with a simple class of problems where values for outcomes are perfectly correlated and derive necessary conditions for optimality of uniform pricing.  Our main theorem later identifies complementary sufficient conditions for general distributions.

\subsection{Perfect Correlations and Necessary Conditions}
Consider a simple class of \emph{\completecorr} instances where the value $t_1$ for outcome~1 pins down the value for outcome~2, $t_2 = \corrcurve(t_1)$.  Assume $\corrcurve(t_1)\leq t_1$, that is, outcome~1 is favored to outcome~2 for all types.   We say that a curve $\corrcurve$ is \emph{\thetamonotone} if $\corrcurve(t_1)/t_1$ is monotone increasing in $t_1$.   Let $\sdpdist$ be the distribution of value for outcome~1.  A distribution $\sdpdist$ is \emph{regular} if its (canonical) amortization of revenue $\sdpamorev(t_1) = t_1 -
\frac{1-\dist_{\max}(t_1)}{\sdpdens(t_1)}$ is monotone non-decreasing in $t_1$ (see the discussion of amortizations of revenue for single-dimensional agents in \autoref{s:amortized}). We investigate optimality of uniform pricing for this class by comparing the profit from a uniform price with the profit from other mechanisms (discounted prices for the less favored outcome or distributions over outcomes).\footnote{With a uniform price when outcome~1 is favored to outcome~2 for all types,  the offer for outcome~2 will not be taken.}  

\begin{theorem}
For any value mapping function $\corrcurve, \corrcurve(t_1)\leq t_1$ that is not \thetamonotone, there exists a regular distribution $\sdpdist$ such that uniform pricing is not optimal for the \completecorr instance jointly defined by $\sdpdist$ and $\corrcurve$.
\label{thm:onlyif}\end{theorem}
\begin{proof}
Let the cost $c=0$. Consider $p$ where ${\corrcurve(t_1)}/{t_1}$ is decreasing at $t_1=p$, and any regular distribution $\sdpdist$ such that $p$ maximizes $p (1-\sdpdist(p))$.  We will show that the revenue of the optimum uniform price $p$ can be improved by another mechanism.  

Consider the change in revenue as a result of supplementing a price $p$ for the outcome~1 with a price $\corrcurve(p) - \epsilon$ for outcome~2.  The results of this change are twofold  (\autoref{fig:completecorr}):  On one hand, a set of types with value slightly less than $p$ for outcome~1 will pay $\corrcurve(p)-\epsilon$ for this new discounted offer. Non-monotonicity at $p$ implies that this set lies above the ray connecting $(0,0)$ to $\corrcurve(p)/p$.  Therefore, for small $\epsilon$ the positive effect is at least
\begin{align*}
(\corrcurve(p)- \epsilon) \times (\sdpdens(p)\cdot  \frac{\epsilon p}{\corrcurve(p)}) = \sdpdens(p) \epsilon p.
\end{align*}

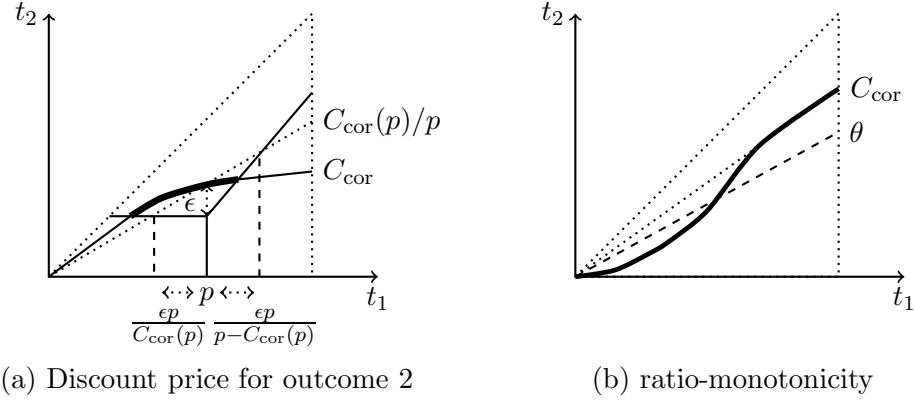
\begin{figure}
\begin{center}

    \begin{tikzpicture}[domain=0:3, scale=3.5, thick]    
    \draw[dotted] (0,0) -- (1,1) -- (1,0);  
    \draw[<->] (0,1) node[left]{$t_{2}$}-- (0,0) -- (1.25,0) node[below] {$t_{1}$};
    \draw[line width=2.5] plot [smooth] coordinates {(.31,.23) (.43,.3) (.6,.35) (.72,.37)};
    \draw (0,0) -- (.31,.23);
    \draw (.72,.37) -- (1,.4) node[right]{$\corrcurve$};
    
    \draw[dotted] (0,0) -- (1,.59) node[right]{$\corrcurve(p)/p$};
    
    \draw[] (.6,0) node[below]{$p$} -- (.6,.23);
    \draw[]  (.6,.23) -- (.23,.23);
    \draw[]  (.6,.23) -- (1,.7);
    \draw[<->,dotted] (.6,.23)  -- (.6,.35);
    \draw (0.6,.275) node[left]{$\epsilon$};
    
    \draw[dashed] (.4,.23) -- (.4,0);
    \draw[dotted,<->] (.43,-0.06) -- (.55,-.06); 
    \draw (.455,-.08) node[below]{$\frac{\epsilon p}{\corrcurve(p)}$};
           
    \draw[dashed] (.8,.45) -- (.8,0);
    \draw[dotted,<->] (.65,-0.06) -- (.78,-.06); 
    \draw (.82,-.08) node[below]{$ \frac{\epsilon p}{p - \corrcurve(p)}$};       
    \draw (0.6,-0.3) node[below]{(a) Discount price for outcome~2};
    
        \draw[<->] (2,1) node[left]{$t_{2}$}-- (2,0) -- (3.25,0) node[below] {$t_{1}$};

       \draw[dotted] (2,0) --(3,1) -- (3,0) -- cycle;
   \draw[dashed] (2,0) -- (3,0.55) node[right]{$\theta$};
       \draw[line width=1.8] plot [smooth] coordinates {(2,0) (2.15,.025) (2.28,.088) (2.35,.13) (2.5,.25) (2.7,.5) (3,5/7)} node[right]{$\corrcurve$};
       \draw[dotted] (2,0) -- (3,5/7);
               \draw (2.6,-0.3) node[below]{(b) \thetamonotonicity};

\end{tikzpicture}
\caption{\footnotesize (a) As a result of adding an offer with price $\corrcurve(p)-\epsilon$ for outcome~2 to the existing offer of price $p$ for outcome~1, the types in darker shaded part of curve $\corrcurve$ will change decisions and contribute to a change in revenue. The lengths of the projected intervals on the $t_1$ axis of the types contributing to loss and gain in revenue are lower- and upper-bounded by $\frac{\epsilon p}{\corrcurve(p)}$ and $ \frac{\epsilon p}{p - \corrcurve(p)}$, respectively. (b) For any $\theta$, $t_1$, $\dist(\theta|t_1,1) = 1$ if $\corrcurve(t_1)/t_1\leq \theta$, and $\dist(\theta|t_1,1) = 0$ otherwise.  Therefore, \thetamonotonicity is equivalent to monotonicity of $\dist(\theta|t_1,1)$ in $t_1$.}
\label{fig:completecorr}
\end{center}
\end{figure}

\noindent On the other hand, a set of types with value slightly higher than $p$ for outcome~1 will change their decision from selecting outcome~1 to outcome~2.  Non-monotonicity at $p$ implies that the negative effect is at most
\begin{align*}
(p - \corrcurve(p) + \epsilon) \times (\sdpdens(p) \cdot  \frac{\epsilon p}{p - \corrcurve(p)}) = \sdpdens(p) \epsilon p.
\end{align*}

\noindent It follows that offering a discount for the less favored outcome strictly improves revenue for small enough $\epsilon$.
\end{proof}

As discussed in \autoref{s:amortized}, a challenge of multi-dimensional mechanism design is that the paths for integration by parts are unknown.  The above theorem highlights another challenge:  even if the paths are known (along the curve for the perfectly correlated class), the incentive compatibility of the mechanism that pointwise optimizes the resulting canonical amortization must be carefully analyzed.  The analysis is a main part of our main theorem in the next section. In contrast to the above theorem, a corollary of our main theorem shows that \thetamonotonicity of $\corrcurve$ and regularity of $\sdpdist$ imply the optimality of uniform pricing.

\subsection{General Distributions and Sufficient Conditions}
We will now state the main theorem of this section which identifies sufficient conditions for optimality of uniform pricing for general distributions.  We say that a distribution over $\typespace \subset \reals^2$ is \emph{\quasisymmetric} if the distribution of maximum value $\val = \max(t_1,t_2)$, conditioned on either $t_1\geq t_2$ or $t_2 \geq t_1$, is identical.\footnote{As examples, any distribution with a domain $\type \in \reals^2, t_1\geq t_2$ is \quasisymmetric (since the distribution conditioned on $t_2\geq t_1$ is arbitrary), as is any symmetric distribution over $\reals^2$.}   Let $\sdpdist(\val)$ and $\sdpdens(\val)$ be the cumulative distribution and the density function of the value for favorite outcome. 
As described in \autoref{s:amortized}, the amortization of revenue for a single-dimensional agent is $\sdpamorev(\val) = \val -
\frac{1-\dist_{\max}(\val)}{\sdpdens(\val)}$.  Let $\dist(\theta|\val,i)$ be the conditional distribution of the \emph{value ratio} $\theta(\type) := \min(t_1,t_2)/\max(t_1,t_2)$ on $\val=t_i\geq t_{-i}$, that is,  $\dist(\theta|\val,i) = \Prx[\type]{\theta(\type)\leq \theta|\val=t_i\geq t_{-i}}$.

\begin{theorem}
\label{t:uniform-pricing-strong-amortization-general} 
Uniform pricing is optimal with $m=2$ outcomes and any service cost $\cost \geq 0$ for 
any \quasisymmetric distribution where (a) the favorite-outcome projection has monotone
non-decreasing amortization of revenue $\phi_{\max}(\val) = \val -
\frac{1-\sdpdist(\val)}{\sdpdens(\val)}$ and (b) the conditional distribution $\dist(\theta|\val,i)$ is monotone non-increasing in $\val$ for all $\theta$ and $i$.
\end{theorem}

Monotonicity of $\dist(\theta|\val,i)$ in $\val$ is correlation of $\theta$ and $\val$ in first order stochastic dominance sense.\footnote{Stronger correlation conditions, such as \emph{Inverse Hazard Rate Monotonicity},  \emph{affiliation}, and independence of favorite value $v$ and the non-favorite to favorite ratio are also sufficient  \citep{MiW82,cas07}.}  It states that as $\val$ increases, more mass should be packed between a ray parameterized by $\theta$, and the 45 degree line connecting $(0,0)$ and $(1,1)$ (\autoref{fig:uniformabove}).  In other words, a higher favorite value makes relative indifference between outcomes, measured by $\theta$, more likely. 

\begin{figure}
\begin{center}

    \begin{tikzpicture}[domain=0:3, scale=3, thick]

        \draw[dotted] (-1.5,0) -- (-.5,0)-- (-0.5,1) -- (-1.5,1);  
    \draw[<->] (-1.5,1.05) node[left]{$t_{2}$}-- (-1.5,0) -- (-0.35,0) node[below] {$t_{1}$};
          \shadedraw [shading = axis, top color = black, bottom color = white] (-1.5,0) -- (-.5,0) -- (-.5,1) -- (-1.5,0);
           \shadedraw [shading = axis, right color = black, left color = white] (-1.5,0) -- (-1.5,1) -- (-.5,1) -- (-1.5,0);

          \draw[dotted] (-1,0) node[below]{$\val \rightarrow$} -- (-1,.5) -- (-1.5,.5);
          \draw (-1,-.1) node[below]{(a)};
          \draw[dashed] (-1.5,0) -- (-.5,.4) node[right]{$\theta$};
          \draw[ ->] (-.8,.1) -- (-.8,.5);
          \draw[->] (-1.4,.7) -- (-1,.7);

       \draw[<->] (0,1.05) node[left]{$t_{2}$}-- (0,0) -- (1.15,0) node[below] {$t_{1}$};

       \draw[dotted] (0,0) -- (1,0) -- (1,1) -- (0,1) -- cycle;
       \draw (1,0) node[below]{$\bar{\val}$};

      \filldraw [line width = 2pt, fill=black!40] plot[id=f3,domain=0.5:1]  (\x,{\x*\x*\x*0.8}) node[right]{$C$}-- (1,1) -- plot[id=f4,domain=1:.5]  ({\x*\x*\x*0.8},\x) -- (.5,.5) -- cycle;
      \draw[dotted] (0,.5) -- (.5,.5)-- (.5,0) node[below]{$\underline{\val}$};
       \draw [line width = 1pt, dotted] plot[id=f3,domain=0:1]  (\x,{\x*\x*\x*0.8});
       \draw [line width = 1pt, dotted] plot[id=f3,domain=0:1]  ({\x*\x*\x*0.8},\x);
       \draw[dotted] (0,0) -- (1,1);
       \draw[dashed] (0,0) -- (1,.5) node[right]{$\theta$};
                     \draw (.5,-0.1) node[below]{(b)};

              \draw (2,-0.1) node[below]{(c)};
        \draw[dotted] (1.5,0) -- (2.5,0) -- (2.5,1) -- (1.5,1);  
\draw (2.5,0) node[below]{$\bar{\val}$};
        \draw[->] (1.5,0) -- (1.5,1.05) node[left]{$t_2$};
        \draw[->] (1.5,0) -- (2.65,0) node[below] {$t_1$};
        \draw[dotted] (1.9,0)  node[below]{$\underline{\val}$} -- (1.9,1);

        \filldraw [line width=2pt, fill=black!40] (1.9,.4) -- (2.5,.4) -- (2.5,1) -- (1.9,1) -- cycle;
        \draw[dashed] (1.5,0) -- (2.5,.7) node[right]{$\theta$};
        \draw[dotted] (1.5+4/7,.4) -- (1.5+4/7,0) node[below]{$v_1$};
                \draw[dotted] (2.3,.56) -- (2.3,0) node[below]{$v_2$};



\end{tikzpicture}
\caption{\footnotesize (a) As $\val$ increases, relatively more mass is packed towards the 45 degree line. (b) A class of distribution satisfying the correlation condition of \autoref{t:uniform-pricing-strong-amortization-general}.  When $t_1\geq t_2$, mass is distributed uniformly above a \thetamonotone curve $C$.  (c) A class of distributions not satisfying the correlation condition of \autoref{t:uniform-pricing-strong-amortization-general}, since $0=\dist(\theta|v_1,1)<\dist(\theta|v_2,1)$.}
\label{fig:uniformabove}
\end{center}
\end{figure}
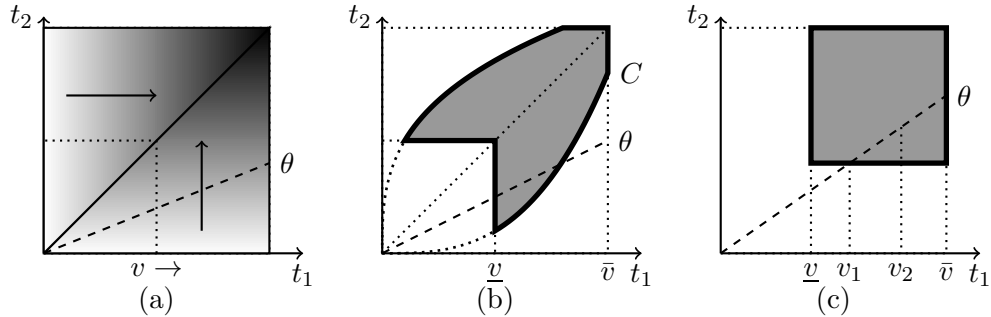

Note the contrast with \autoref{thm:onlyif}.  For a \completecorr instance, $\dist(\theta|\val,1)=1$ if $\corrcurve(\val)/\val \leq \theta$, and $\dist(\theta|\val,1)=0$ otherwise (\autoref{fig:completecorr}).  Monotonicity of $\dist(\theta|\val,1)$ in $\val$ is therefore equivalent to \thetamonotonicity of $\corrcurve$.  \autoref{t:uniform-pricing-strong-amortization-general}  states that for any \thetamonotone $\corrcurve$, uniform pricing is optimal for the perfectly correlated instance jointly defined by $\corrcurve$ and \emph{any} regular distribution $\sdpdist$.  As another class of distributions satisfying the conditions of \autoref{t:uniform-pricing-strong-amortization-general}, one can draw the maximum value $v$ from a regular distribution $\sdpdist$, and the minimum value uniformly from $[\curve(\val),\val]$, for a \thetamonotone function $\curve$ satisfying $\curve(\val)\leq \val$ (\autoref{fig:uniformabove}).  On the other hand, a distribution where values for outcomes are uniformly and independently drawn from $[\underline{\val},\bar{\val}]$, with $\underline{\val}>0$, does not satisfy the conditions (when $\underline{\val} = 5, \bar{\val}=6$, \citealp{Tha04}, showed that uniform pricing is not optimal).  As another example, if $t_1$ and $t_2$ are drawn independently from a distribution with density proportional to $e^{h(\log(x))}$ for any monotone non-decreasing convex function $h$, then the distribution satisfies the conditions of the theorem (see \aref{app:independent}).

Notice that the conditional distributions $\dist(\theta|\val,i)$ jointly with $\sdpdist$ are alternative representations of any \quasisymmetric distribution as follows: with probability $\Prx{t_1\geq t_2}$, draw $t_1$ from $\sdpdist$, $\theta$ from $\dist(\cdot|t_1,1)$, and set $t_2 = t_1\theta$ (otherwise assign favorite value to $t_2$ and draw $\theta$ from $\dist(\cdot|t_2,2)$). As a result, the requirements of \autoref{t:uniform-pricing-strong-amortization-general} on $\sdpdist$ and $\dist(\theta|\val,i)$ are orthogonal.  This view is particularly useful since it is natural to define several instances of the problem in terms of distributions over parameters $\val$ and $\theta$.  For example, in the pricing with delay model discussed in the introduction, $\theta$ has a natural interpretation as the discount factor for receiving an item with delay.  We will revisit the conditions of \autoref{t:uniform-pricing-strong-amortization-general} in \autoref{s:extensions}.

The rest of this section proves the above theorem by constructing the appropriate virtual value functions.     Notice that \quasisymmetry allows us to focus on only the conditional distribution when the favorite outcome is outcome~1.  If a single mechanism, namely uniform pricing, is optimal for each case (of outcome~1 or outcome~2 being the favorite outcome), the mechanism is optimal for any probability distribution over the two cases. Therefore for the rest of this section we work with the distribution conditioned on $t_1\geq t_2$. In particular, $T$ is a bounded subset of $\reals^2$ specified by an interval $[\underline{t}_1,\bar{t}_1]$ of values $t_1$ and bottom and top boundaries $\underline{t}_2(t_1)$ and $\bar{t}_2(t_1)$ satisfying $\underline{t}_2(t_1) \leq \bar{t}_2(t_1)\leq t_1$.  The proof follows the framework of \autoref{s:amortized}.    In \autoref{d:2d-extension} we define $\amorev$ and $\amortil$ from the properties they must satisfy to prove optimality of uniform pricing.  \autoref{lem:strongamort} shows that $\amorev$ is a canonical amortization and is tight for any uniform pricing.  \autoref{l:2d-ext=>uniform-pricing} shows that given the conditions of \autoref{t:uniform-pricing-strong-amortization-general} on the distribution, the allocation of uniform pricing maximizes virtual surplus pointwise with respect to $\amorev$. The theorem follows from \autoref{prop:vvf=>IC+opt}.

A uniform pricing $p \in[\underline{t}_1,\bar{t}_1]$ implies $\alloc(\type)=0, \util(\type)=0$ if $t_1\leq p$, and $\alloc(\type)=(1,0), \util(\type)>0$ otherwise (recall the assumption that $t_1\geq t_2$).  Therefore, in order to satisfy the requirement of \autoref{lem:framework} that $\util(\type)(\amortil \cdot \normal)(\type)  =0$ everywhere on the boundary and for all uniform pricings $p \in[\underline{t}_1,\bar{t}_1]$, $\amortil$ must be \emph{boundary orthogonal}, $(\amortil \cdot \normal)(\type)= 0$, except possibly at the left boundary, where $\util(\type) = 0$ (\autoref{fig:partialorthogonality}).  With this refinement of \autoref{lem:framework} of the boundary conditions of $\amortil$ we now define $\amortil$ and $\amorev$.

\begin{figure}
\centering
\begin{tikzpicture}[scale = 3.5, align=center]
    \draw[<->] (0,1) node[left]{$t_{2}$}-- (0,0) -- (1.15,0) node[below] {$t_{1}$};

\draw[dotted] (0,0) -- (1,1);
\filldraw[line width=1.8, black!40] plot [smooth] coordinates {(.4,.1) (.6,.05) (1,.3)} -- (1,.5) -- (.4,.3) -- cycle;
\draw[line width=1.8] plot [smooth] coordinates {(.4,.1) (.6,.05) (1,.3)} -- (1,.5) -- (.4,.3) -- cycle;
\filldraw (1,.4) circle(.5pt) node[right]{$\sdamortil_1 = 0$};
\filldraw (.7,.4) circle(.5pt) node[above,yshift=5pt]{$\amortil \cdot \normal = 0$};
\filldraw (.4,.2) circle(.5pt) node[left]{$\sdamortil_1 \geq 0$};
\filldraw (.7,.1) circle(.5pt) node[below,xshift=20pt,yshift=4pt]{$\amortil \cdot \normal = 0$};
    \draw (0.5,-0.1) node[below]{(a) Boundary conditions of $\amortil$};
    
    \filldraw[black!40] plot [smooth] coordinates {(2,.1) (2.3,.025) (2.9,.4)} -- (2.9,.7) -- (2,.4) -- cycle;
    \draw[line width=2.5] plot [smooth] coordinates {(2,.1) (2.3,.025) (2.9,.4)};
    \draw[line width=1] (2.9,.4)--(2.9,.7) -- (2,.4) -- (2,.1);
\filldraw (2.3,.03) circle(.7pt) node[below,xshift=40pt]{$\amortil \cdot \normal = 0 \Rightarrow \sdamortil_2 = -\sdamortil_1\sdnormal_1/\sdnormal_2$};
\draw[line width = 1.5] (2.3,.03) -- (2.3,.3);
\filldraw (2.3,.3) circle(.5pt) node[above,xshift=17pt]{$\Delta \sdamortil_2 = \int \partial_2 \sdamortil_2$};
    \draw (2.5,-0.1) node[below]{(b) Definition of $\sdamortil_2$};

\end{tikzpicture}
\caption{\footnotesize (a) In addition to the divergence density equality $\nabla \cdot \amortil = -f$, $\amortil$ must be boundary orthogonal at all boundaries except possibly the left boundary, and an inflow at the left boundary. (b) Given $\sdamortil_1$, we solve for $\sdamortil_2$ to satisfy boundary orthogonality at the bottom and divergence density equality.  Boundary orthogonality uniquely defines $\sdamortil$ on the bottom boundary.  Integrating the divergence density equality $\partial_2 \sdamortil_2 = -\dens - \partial_1\sdamortil_1$ defines $\sdamortil_2$ everywhere.}
\label{fig:partialorthogonality}\end{figure}
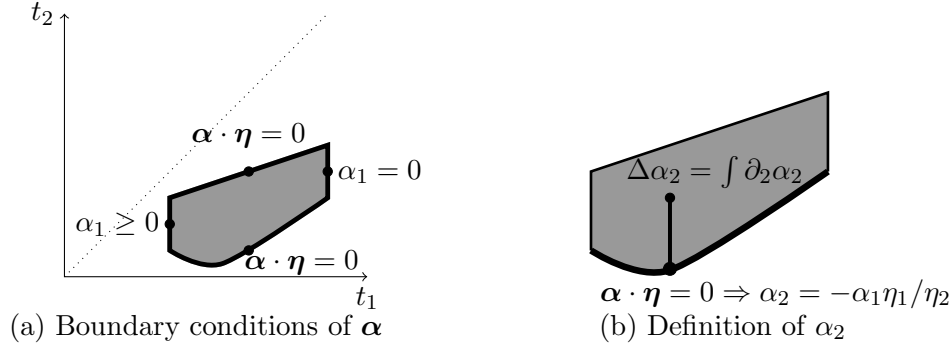

\begin{definition} 
\label{d:2d-extension}
The \emph{two-dimensional extension} $\amorev$ of the amortization for the favorite-outcome
projection $\sdpamorev(\val) = \val -
\frac{1-\sdpdist(\val)}{\sdpdens(\val)}$ is constructed as follows:
\begin{enumerate}[(a)]
\item Set $\sdamorev_1(\type) = \sdpamorev(t_1)$ for all $\type \in \typespace$.
\item Let $\sdamortil_1(\type) = (t_1 -
\sdamorev_1(\type))\,\dens(\type) = \frac{1-\sdpdist(t_1)}{\sdpdens(t_1)} \, \dens(\type)$.\footnote{Our assumption that $\dens>0$ and the regularity assumptions on $\typespace$ imply that $\sdpdens>0$ everywhere except potentially at the left boundary if the left boundary is a singleton.  We treat this case separately in the upcoming proof of the theorem.}
\item Define $\sdamortil_2(\type)$ uniquely to satisfy divergence density equality $\partial_2 \sdamortil_2 = -f - \partial_1\sdamortil_1$ and boundary orthogonality of the bottom boundary.
\item Set $\sdamorev_2(\type) = t_2 - \sdamortil_2(\type) / \dens(\type)$.
\end{enumerate}
\end{definition}

An informal justification of the steps of the construction is as
follows:
\begin{enumerate}[(a)]
\item \label{phi1tight} First, $\sdamorev_1(\type)$ may only be a function of $t_1$; otherwise, if $\sdamorev_1(\type)>\sdamorev_1(\type')$ with $t_1=t'_1$, maximizing virtual surplus pointwise with cost $\cost$ satisfying $\sdamorev_1(\type)>c>\sdamorev_1(\type')$ implies $x_1(\type') = 0$, and either $x_1(\type)>0$ or $x_2(\type)>0$ (if $\sdamorev_2(\type)>\sdamorev_1(\type)>c$).  Such an allocation $\alloc$ is not the allocation of uniform pricing.\footnote{This argument applies only if $\sdamorev_1(\type)>0$.  Nevertheless, we impose the requirement that $\sdamorev_1(\type) = \sdamorev_1(t_1)$ everywhere as it allows us to uniquely solve for $\amorev$.}  Second, given the first point, the expected virtual surplus of uniform pricing $p$ is $\int_{t_1\geq p} [\sdamorev_1(t_1) \sdpdens(t_1) -c] \; \dd t_1$, which by tightness we need to be equal to $(p-c)(1-\sdpdist(p))$.  Solving this equation for all $p$ gives $\sdamorev_1(\type) = \sdpamorev(t_1)$.
\item 
\label{step:amorev-definition-alpha1}
We obtain $\sdamortil_1$ from $\sdamorev_1$ by \autoref{def:amortiltoamorev}.  
\item \label{step:amorev-definition-alpha2} Given $\sdamortil_1$, $\sdamortil_2$ is defined to satisfy divergence density equality, $\partial_2 \sdamortil(\type) = -\dens(\type) - \partial_1 \sdamortil(\type)$, and boundary orthogonality at the bottom boundary (i.e., $t_2=\underline{t}_2(t_1)$).  Integrating and employing boundary orthogonality
  on the bottom boundary of the type space, which requires that $\amortil \cdot \normal =0$,  gives the
  formula (\autoref{fig:partialorthogonality}).  For example, if $\underline{t}_2(t_1)=0$, boundary orthogonality requires that $\sdamortil_2(t_1,0) = 0$, and thus $\sdamortil_2(\type) = - \int_{y=0}^{t_2} \left(\dens(t_1,y) + \partial_1 \sdamortil_1(t_1,y)\right)   \dd y$.
\item We obtain $\sdamorev_2$ from $\sdamortil_2$ by
  \autoref{def:amortiltoamorev}.
\end{enumerate}

For $\amorev$ to prove optimality of uniform pricing, we need the allocation of uniform pricing to optimize virtual surplus pointwise with respect to $\amorev$.  This additional requirement demands
that $\sdamorev_1(\type) \geq \sdamorev_2(\type)$ for any type $\type
\in \typespace$ for which either $\sdamorev_1(\type)$ or $\sdamorev_2(\type)$ is positive.  A little algebra shows that this condition is
implied by the angle of $\amortil(\type)$ being at most the angle
of $\type$ with respect to the horizontal $t_1$ axis, that is, $t_2\sdamortil_1(\type)\leq t_1\sdamortil_2(\type)$ (\autoref{l:2d-ext=>uniform-pricing}, below).  The direction of
$\amortil$ corresponds to the paths on which incentive compatibility
constraints are considered.  Importantly, our approach does not fix
the direction and allows any direction that satisfies the above
constraint on angles.  The following lemma is proved by the divergence
theorem, and specifies the direction of $\amortil$.

\begin{definition}
For any $q\in [0,1]$, define the \emph{equi-quantile} function $\curve_q(t_1)$ such that conditioned on $t_1$, the probability that $t_2\leq  \curve_q(t_1)$ is equal to $q$ (see \autoref{fig:trapezoidalquantile}).  More formally, $\curve_q$ is the upper boundary of $T_q$, where
\begin{align*}
T_q = \{\type | \Prx[\type']{t'_2\leq t_2|t'_1=t_1, t'_1 \geq t'_2}:= \frac{\int_{t'_2\leq t_2} \dens(t_1,t'_2)\dd t'_2}{\int_{t'_2\leq t_1} \dens(t_1,t'_2)\dd t'_2} \leq q\}.
\end{align*}
\end{definition}

For example, notice that for the \completecorr class, the equi-quantile curves $C_q$ are identical to $\corrcurve$.

\begin{lemma}
The vector field $\amorev$ of \autoref{d:2d-extension} is a tight canonical amortization for any uniform pricing. At any $\type$, $\amortil(\type)$ is tangent to the equi-quantile curve crossing $\type$.
\label{lem:strongamort}\end{lemma}
\begin{proof}

Tightness follows directly from the definition of $\sdamorev_1$ (see the justification for Step~\eqref{phi1tight} of the construction).  The divergence density equality and bottom
boundary orthogonality of $\amortil$ are automatically satisfied by
Step~\eqref{step:amorev-definition-alpha2} of the construction.  Orthogonality of the right
  boundary ($t_1=\bar{t}_1$) requires that $\amortil(\bar{t}_1,t_2) \cdot (1,0) = 0$, which is $\amortilelement_1(\bar{t}_1,t_2) = 0$. This property follows directly from the definitions since $\sdamorev_1(\bar{t}_1,t_2) = \sdpamorev(\bar{t}_1) = \bar{t}_1$, and therefore $\sdamortil_1(\bar{t}_1,t_2) = (\bar{t}_1 -
\sdamorev_1(\bar{t}_1,t_2))\,\dens(\bar{t}_1,t_2) = 0$.  At the left boundary, $\amortil \cdot \normal \leq 0$ since $\sdamortil_1 \geq 0$ from definition and the normal vector is $(-1,0)$. The only remaining condition, the top boundary orthogonality, is implied by the tangency property of the lemma as follows.  The top boundary is $C_1$. Tangency of $\amortil$ to $C_1$ implies that $\amortil$ is orthogonal to the normal, which is the top boundary orthogonality requirement.  It only remains to prove the tangency property.

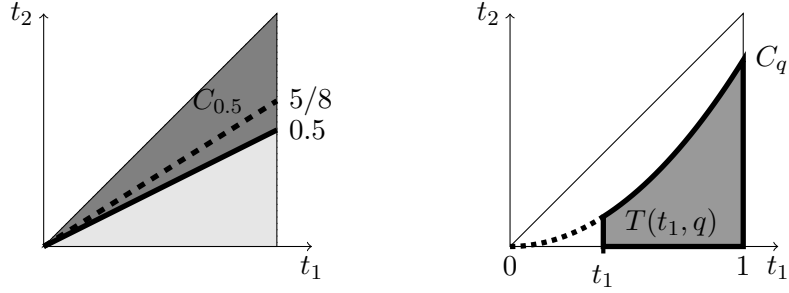
\begin{figure}
\centering
\begin{tikzpicture}[scale = 3.1, align=center]

    \draw[dotted] (0,0) -- (1,1) -- (1,0);  
    \draw[<->] (0,1) node[left]{$t_{2}$}-- (0,0) -- (1.15,0) node[below] {$t_{1}$};

\filldraw [fill=black!10] (0,0) -- (1,.5) -- (1,0) -- cycle;
\filldraw [fill=black!50] (0,0) -- (1,.5) -- (1,1) -- cycle;
\draw[line width=2] (0,0) -- (1,.5) node[right]{$0.5$};

\draw[line width = 2, dashed] (0,0) -- (1,5/8) node[right]{$5/8$};
\draw (.9,5/8) node[left]{$C_{0.5}$};

\draw (2,0) -- (3,0) -- (3,1) -- (2,0);
    \draw[<->] (2,1) node[left]{$t_{2}$}-- (2,0) -- (3.15,0) node[below] {$t_{1}$};

\draw (2,0) node[below]{$0$};
\draw (3,0) node[below]{$1$};
\draw[ultra thick] (2.4,0) -- (2.4,-0.03) node[below]{$t_1$};
\filldraw [line width = 2pt, fill=black!40] (2.4,0.096)  -- plot[id=f3,domain=0.4:1]  (2+\x,{\x*\x*0.8}) node[right]{$C_q$}-- (3,0) -- (2.4,0) -- cycle;
\draw[dotted, line width = 2pt] plot[id=f3,domain=0:.4]  (2+\x,{\x*\x*0.8});
\draw (2.45,0.1) node[right]{$T(t_1,q)$};
\end{tikzpicture}
\caption{\footnotesize (a) The density in the darker region is twice the density in lighter region.   For example, $C_{0.5}(t_1) = 5t_1/8$, meaning given $t_1$, the probability that $t_2\leq 5t_1/8$ is $1/2$. (b) $T(t_1,q)$ is the set of types below $C_q$ and to the right of $t_1$. The four curves that define the boundary of $T(t_1,q)$ are $\{\text{TOP, RIGHT, BOTTOM, LEFT}\}(t_1,q)$.  For simplicity the picture assumes $T$ is the triangle defined on (0,0), (1,0), and (1,1).}
\label{fig:trapezoidalquantile}\end{figure}

The strategy for the proof of the tangency property is as follows.  We fix
$t_1$ and $q$ and apply the divergence theorem to $\amortil$ on
the subspace of type space to the right of $t_1$ and below $C_q$.\footnote{The divergence theorem for vector field $\vectorfield$ is $\int_{\typeprelim\in T} (\nabla \cdot \vectorfield)(\typeprelim) \; \dd \typeprelim =  \int_{\typeprelim \in \partial T} (\vectorfield \cdot \boldsymbol{\eta})(\typeprelim) \; \dd \typeprelim.$} More formally, divergence theorem is applied to the set of types $T(t_1,q) = \{\type'\in \typespace|t'_1 \geq t_1; \dist(t_2|t_1) \leq q\}$ (see \autoref{fig:trapezoidalquantile}). The divergence theorem equates the integral of the
orthogonal magnitude of vector field $\amortil$ on the boundary of the subspace to the
integral of its divergence within the subspace.  As the upper boundary
of this subspace is $C_q$, one term in this
equality is the integral of $\amortil(\type')$ with the upward
orthogonal vector to $C_q$ at $\type'$.  Differentiating this integral with respect to $t_1$ gives the desired quantity.
\begin{align}
&\hspace{-6.8ex}\int_{\type' \in \text{TOP}(t_1,q)}\normal(\type') \cdot \amortil(\type') \, \dd \type' \nonumber\\
&= \int_{\type' \in T(t_1,q)} \nabla \cdot \amortil(\type') \, \dd \type'
-\int_{\type' \in \{\text{RIGHT,BOTTOM,LEFT}\}(t_1,q)}  \normal(\type') \cdot \amortil(\type') \, \dd \type'.\label{eq:div}\\
\intertext{Using divergence density equality and boundary orthogonality the right hand side becomes}
&= -\int_{\type' \in T(t_1,q)} \dens(\type') \, \dd \type'
-\int_{\type' \in \{\text{LEFT}\}(t_1,q)}  \normal(\type') \cdot \amortil(\type') \, \dd \type' \nonumber\\
&= -q(1-\sdpdist(t_1))
-\int_{\type' \in \{\text{LEFT}\}(\quantile)}  \normal(\type') \cdot \amortil(\type') \, \dd \type' \nonumber
\end{align}
\noindent where the last equality followed directly from definition of $T(t_1,q)$. By definition of $\amortil$, and since normal $\normal$ at the left boundary is $(-1,0)$,
\begin{align*}
\int_{\type' \in \{\text{LEFT}\}(t_1,q)}  \normal(\type') \cdot \amortil(\type') \, \dd \type'&=-\frac{1-\sdpdist(t_1)}{\sdpdens(t_1)}\int_{t'_2\leq C_q(t_1)} \dens(t_1,t'_2)\; \dd t'_2\\
&= -\frac{1-\sdpdist(t_1)}{\sdpdens(t_1)} q \sdpdens(t_1) \\
&= - (1-\sdpdist(t_1)) q.
\end{align*}
As a result, the right-hand side of equation~\eqref{eq:div} sums to zero, and we have
\begin{align*}
\int_{\type' \in \text{TOP}(t_1,q)}\normal(\type') \cdot \amortil(\type') \, \dd \type' = 0.
\end{align*}
\noindent Since the above equation must hold for all $t_1$ and $q$, we conclude that $\amortil$ is tangent to the equi-quantile curve at any type.
\end{proof}

The following lemma gives sufficient conditions for uniform pricing to be the pointwise maximizer of virtual surplus given any cost $c$. These conditions imply that whenever $\sdamorev_1(\type)\geq c$ then $\sdamorev_1(\type)\geq\sdamorev_2(\type)$, and that $\sdamorev_1(\type)\geq c$ if and only if $t_1$ is greater than a certain threshold (implied by monotonicity of $\sdamorev_1(\type)\geq c$).

\begin{lemma}
\label{l:2d-ext=>uniform-pricing}
The allocation of a uniform pricing mechanism optimizes virtual surplus pointwise with respect to $\amorev = \type - \amortil/\dens$ of \autoref{d:2d-extension} and any
non-negative service cost $\cost$ if the equi-quantile curves are \thetamonotone and $\sdamorev_1(\type)$ is monotone non-decreasing in $\type$.
\end{lemma}

\begin{proof}
Tangency of $\alpha$ to the equi-quantile curves (\autoref{lem:strongamort}) implies that $\frac{t_2}{t_1} \,
\amortilelement_1(t_1,t_2) - \amortilelement_2(t_1,t_2) \leq 0$ if all equi-quantile curves are \thetamonotone. From the assumption $\frac{t_2}{t_1} \, \amortilelement_1(t_1,t_2) -
\amortilelement_2(t_1,t_2) \leq 0$ and \autoref{d:2d-extension} we
have
\begin{align*}
\frac{t_2}{t_1} \, \sdamorev_1(\type) 
  &= \frac{t_2}{t_1}\,\big(t_1 - \frac{\sdamortil_1(\type)}{\dens(\type)}\big) 
  = t_2 - \frac{t_2}{t_1}\cdot \frac{\sdamortil_1(\type)}{\dens(\type)}
  \geq t_2 - \frac{\sdamortil_2(\type)}{\dens(\type)} 
  = \sdamorev_2(\type).
\end{align*}
Thus, for $\type$ with $\sdamorev_1(\type) \geq \cost$,
$\sdamorev_1(\type) \geq \sdamorev_2(\type)$ and pointwise virtual surplus
maximization serves the agent outcome~1.  Since $\sdamorev_1(\type)$ is
a function only of $t_1$ (\autoref{d:2d-extension}), its monotonicity implies that there is a smallest $t_1$ such that all greater types are served. Also, if $\sdamorev_1(\type) \leq c$, again the above calculation implies that $\sdamorev_2(\type) \leq c$ and therefore the type is not served.
This allocation is the allocation of a uniform pricing.
\end{proof}

\begin{proof}[Proof of \autoref{t:uniform-pricing-strong-amortization-general}]  We show that $\amorev = \type - \amortil/\dens$ of \autoref{d:2d-extension} is a virtual value function for a uniform pricing and invoke \autoref{prop:vvf=>IC+opt}. \autoref{lem:strongamort} showed that $\amorev$ is a tight amortization for any uniform pricing.\footnote{Special attention is needed in case that the left boundary is a singleton, since in that case $\sdpdens(\underline{t}_1)) = 0$ and $\sdamortil_1$ is unbounded.  In this case our analysis showed that $\amortil \cdot \normal = 0$ everywhere except possibly at $(\underline{t}_1,t_2(\underline{t}_1))$.  Divergence theorem states that \begin{align} 
\label{eq:boundary-integral}
\int_{\type \in \typeboundary}(\amortil \cdot \normal)(\type)\,
\dd\type &= - \int_{\type \in \typespace} \dens(\type)\,\dd \type = -1,
\end{align}
which implies that $\amortil \cdot \normal$ is a negative Dirac delta at $(\underline{t}_1,t_2(\underline{t}_1))$.  The integral of  $\util(\amortil \cdot \normal)$ over the boundary is thus $-\util(\underline{t}_1,t_2(\underline{t}_1))=0$.} \autoref{l:2d-ext=>uniform-pricing} showed that the allocation of a uniform pricing pointwise maximizes virtual surplus with respect to $\amorev$.
\end{proof}

\subsection{Extensions}\label{s:extensions}
This section contains extensions of \autoref{t:uniform-pricing-strong-amortization-general} to $m\geq 2$ outcomes, $n\geq 1$ agents, and distributions where the favorite-outcome projection may not be regular.

First, \autoref{t:uniform-pricing-strong-amortization-general} can be extended to the case of more than two outcomes and more than one agent.  The positive correlation property becomes a sequential positive correlation where the ratio of the value of any outcome to the favorite outcome is positively correlated with the value of favorite outcome, conditioned on the draws of the lower indexed outcomes. A distribution over types $[0,1]^m$ is \quasisymmetric if the distribution of $\val = \max_i t_i$ stays the same conditioned on any outcome having the highest value.   For $j \neq i$, define $q^i_j(\type)$ to be the quantile of the distribution of $t_j$ conditioned on $i$ being the favorite outcome, and conditioned on the values $\type_{<j} = (t_1,\ldots,t_{j-1})$ of the lower indexed outcomes.  Formally, $q^i_j(\type) = \Prx[\type']{t'_j\leq t_j|\type'_{< j}=\type_{< j}, t'_i=t_i=\max_k t'_k}$.  Define $\dist(\theta_j|t_i,i, \boldsymbol{q}_{<j}) = \Prx[\type']{t'_j/t'_i\leq \theta_j | \boldsymbol{q}_{<j} = \boldsymbol{q}^i_{<j}(\type'), t'_i=t_i=\max_k t'_k}$ to be the distribution of the value ratio of $j$th to favorite outcome, conditioned on $i$ being the favorite outcome and given vector $\boldsymbol{q}_{<j}$ of the quantiles of the lower indexed outcomes.  In the multi-agent problem with a configurable item, a single item with $m$ configurations is to be allocated to at most one of the agents.\footnote{We assume that the item has the same possible configurations for each agent.  This can be achieved by defining the set of configurations to be the union over the configurations of all agents.}

\begin{theorem}
\label{t:uniform-pricing-strong-amortization-many-agents-items} 
A favorite-outcome projection mechanism is optimal for an item with $m\geq1$ configurations, multiple independent agents, and any service cost $\cost \geq 0$, if the distribution of each agent is \quasisymmetric and (a) the favorite-outcome projection has monotone
non-decreasing amortization $\phi_{\max}(v) = v -
\frac{1-\sdpdist(v)}{\sdpdens(v)}$ and (b) $\dist(\theta_j|\val,i, \boldsymbol{q}_{<j})$ is monotone non-increasing in $\val$ for all $i$, $j$, $\theta_j$, and $\boldsymbol{q}_{<j}$.
\end{theorem}

The proof of the above theorem is in \aref{a:proofmostgeneral}.  From \citet{M81} we know that if a favorite-outcome projection mechanism is optimal, the optimum mechanism is to allocate the item to the agent with highest $\phi_{\max}(v)$ (no ironing is required as we are assuming regularity), and let the agent choose its favorite configuration.  With a single agent, the configurable item setting is identical to the original model with multiple outcomes.  The above theorem implies it is optimal to offer a single agent a price for its choice of outcome, generalizing \autoref{t:uniform-pricing-strong-amortization-general} to $m\geq 2$ outcomes.  A special case of the correlation above is when the ratios are independent of each other conditioned on the value of the favorite outcome, that is, each $\theta_j = t_j/\val$ for $j\neq i$ is drawn independently of others from a conditional distribution $\dist(\theta|\val,i)$ that is monotone in $\val$.

The second extension removes the regularity assumption of \autoref{t:uniform-pricing-strong-amortization-many-agents-items} by assuming a slightly stronger correlation assumption, and designs a virtual value function with a simple sweeping procedure  in a single dimension (proof in \aref{a:ironing}).  In particular, we only iron the canonical amortization $\amorev$ along the equi-quantile curves.

\begin{theorem}
\label{thm:unitdemandironing}
A favorite-outcome projection mechanism is optimal for an item with $m=2$ configurations, multiple independent agents, and any service cost $\cost \geq 0$, if the distribution of each agent is \quasisymmetric with convex equi-quantile curves.
\end{theorem}

From \citet{M81}, optimality of a favorite-outcome projection mechanism implies optimality of allocating to the agent with highest ironed virtual value. \autoref{fig:convexitymonotonicity} depicts how convexity of equi-quantile curves is stronger than the stochastic dominance requirement of \autoref{t:uniform-pricing-strong-amortization-general}. Convexity states that the line connecting any two points, namely $(0,0)$ and $(t_1,t_1\theta)$, lies above the curve for all $t'_1\leq t_1$, and below the curve for all $t'_1\geq t_1$. As a result, for any $t'_1\geq t_1$, $\dist(\theta|t'_1)\leq \dist(\theta|t_1)$, and the other direction holds for $t'_1\leq t_1$ (see \autoref{fig:convexitymonotonicity}).

\begin{figure}
\begin{center}
    \begin{tikzpicture}[domain=0:3, scale=2.7, thick]    
    \draw[dotted] (0,0) -- (1,1) -- (1,0);  
    \draw[<->] (0,1.2) node[left]{$t_{2}$}-- (0,0) -- (1.25,0) node[below] {$t_{1}$};
    \draw[line width=2] plot [smooth] coordinates {(0,0) (.28,.08) (.37,.11) (.54,.2) (.7,.35) (1,.7)} node[right]{$q$};
    \draw[dashed] (0,0) -- (1,0.5) node[right]{$\theta$};
    \draw[dotted] (.7,.35) -- (0.7,0);
    \draw[thick] (0.7,.0) -- (0.7,-0.03) node[below]{$t_1$};
                            \draw (.5,-.2) node[below]{(a) convex equi-quantile curve};

    \draw[dotted] (2.5,0) -- (3.5,1) -- (3.5,0);  
    \draw[<->] (2.5,1.2) node[left]{$t_{2}$}-- (2.5,0) -- (3.75,0) node[below] {$t_{1}$};
    \draw[dashed] (2.5,0) -- (3.5,.75);
    
    \draw[line width=2] (2.5,0) -- (2.85,.1) -- (3.1,.45) -- (3.5,.75);
                                \draw (3,-.2) node[below]{(a) non-convex monotone equi-quantile curves};
\end{tikzpicture}
\caption{\footnotesize The connection between convexity and \thetamonotonicity of equi-quantile curves. (a) Convexity implies \thetamonotonicity. (b) \thetamonotonicity does not imply convexity.}
\label{fig:convexitymonotonicity}
\end{center}
\end{figure}
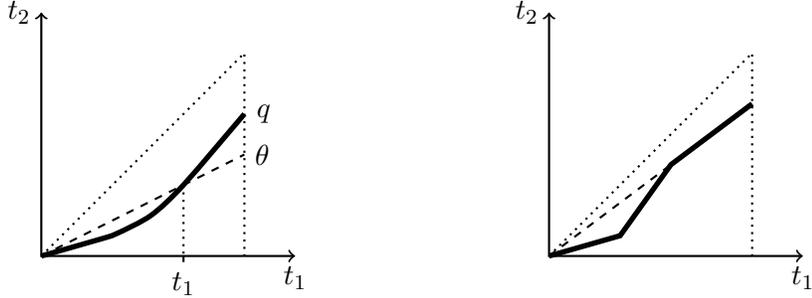

\section{Grand Bundle Pricing for Additive Preferences}\label{a:add bundle}
In single-agent multi-product settings with free disposal (i.e., value for a set of items does not decrease as more items are added), optimality of a favorite-outcome projection mechanism is equivalent to optimality of posting a single price for the grand bundle of items.  Thus, \autoref{t:uniform-pricing-strong-amortization-many-agents-items} can be used to obtain conditions for optimality of grand bundle pricing.  For example, in the case of two items, when the value for the bundle is $\val$ and value for individual items are $\val \delta_1$ and $\val \delta_2$,  \autoref{t:uniform-pricing-strong-amortization-many-agents-items} identifies a sufficient positive correlation condition.  Note that the theorem does not require any structure on values, such as additivity  (value for a bundle is the sum of the values of items in it) or super- or sub-additivity, other than free disposal.  If the preference is indeed additive, we have $\delta_2 = 1-\delta_1$, and \autoref{t:uniform-pricing-strong-amortization-many-agents-items} requires that $\delta_1$ be both positively and negatively correlated with $\val$.  The only admissible case is independence.\footnote{Let $t_1=\val$ be the value for the bundle, and $t_2=\delta \val$ and $t_3 = (1-\delta)\val$ the values for the two items.  Let $\delta(q,\val)$ be the inverse of the quantile mapping, i.e., $\Prx{\delta\leq \delta(q,\val)|\val} = q$.  \autoref{t:uniform-pricing-strong-amortization-many-agents-items} demands that $\delta(q,\val)$ be monotone non-decreasing and $\dist(1-\delta\leq \theta_2|\val, \delta = \delta(q,\val))$ be monotone non-increasing in $\val$ for all $q, \theta_2$.  The only possible case is independence of $\val$ and $\delta$, that is, $\delta(q,\val)$ is a constant.}  In this section we apply the framework of \autoref{s:amortized} to prove optimality of grand bundle pricing for additive preferences, and obtain conditions of optimality that are more permissive than  independence by constructing a virtual value function $\virt$ from a canonical amortization $\amorev$ that is tight for any grand bundle pricing and is constructed to satisfy conditions of \autoref{lem:framework}.  In this section we consider a single agent, $m=2$ items.  As discussed in \autoref{sec:reverseengineering} and similar to \autoref{s:uniform}, we use a class of cost functions to restrict the admissible amortizations.  In particular,  we assume that the cost of an allocation $\alloc\in [0,1]^2$ is $\cost(\alloc) = \cost \max(x_1,x_2)$ for a $\cost \geq 0$.

Similar to \autoref{s:uniform}, we first study a family of instances with perfect correlation to obtain necessary conditions of optimality.  In particular,  let $\sumdist$ be a distribution over value $s$ for the bundle (in the case of two items we refer to the grand bundle simply as the bundle), and $\theta(s)$ be the ratio of the value of item~2 to item~1 when value for the bundle is $s$, that is, value for item~1 is $t_1 = s/(1+\theta)$, and value for item~2 is $t_2 = \theta s/(1+\theta)$.\footnote{Because of the additivity structure imposed on preferences, two parameters are sufficient to define values for three outcomes.  For example, $t_1$ and $t_2$ define the value for the bundle $s=t_1+t_2$.  Alternatively, $s$ and $\theta$ define the value for individual items.} The following theorem shows that if $\theta(s)$ is not monotone non-increasing in $s$, then bundling is not optimal for some distribution $\sumdist$.  The proof is similar to \autoref{thm:onlyif} and is omitted.

\begin{theorem}\label{thm:onlyifadditive}
If $\theta(s)$ is not monotone non-increasing in $s$, then there exists a regular distribution $\sumdist$ over $s$ such that grand bundle pricing is not optimal for the \completecorr instance jointly defined by $\sumdist$ and $\theta(\cdot)$ and with zero costs.
\end{theorem}

The main theorem of this section states sufficient conditions for optimality of pricing the bundle.  A symmetric distribution is identified by a marginal distribution $\sumdist$ of value for the bundle $s$ as well as a conditional distribution $F(\theta|s)$ of the ratio $\theta(\type) = \max(t_1,t_2)/\min(t_1,t_2)$ conditioned on value for the bundle $s$.  The main theorem of this section states that regularity of $\sumdist$ and negative correlation of $s$ and $\theta$ in the first order stochastic dominance sense is sufficient for optimality of bundling.

\begin{theorem}\label{thm:additiveif}
For a single agent with additive preferences over two items, bundle pricing is optimal for any costs $\cost \max(x_1,x_2)$, $c \geq 0$, and any symmetric distribution where (a) $\sumdist$ has monotone amortization $\sumamorev$ and (b) the conditional distribution $F(\theta|s)$ is monotone non-decreasing in $s$.
\end{theorem}

The following is an example class of distributions satisfying the conditions of \autoref{thm:additiveif}. Draw the value for the bundle $s$ from a regular distribution $\sumdist$, and value for the items $t_1$ and $t_2$ uniformly such that $t_1+t_2=s$, $\max(t_1,t_2)/\min(t_1,t_2)\geq \theta(s)$, for any monotone non-increasing function $\theta(s)$ (see \autoref{fig:exampleadditive}).

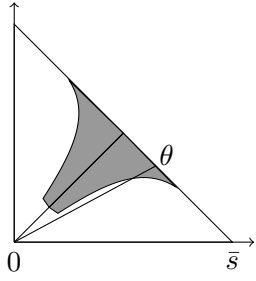
\begin{figure}
\centering
\begin{tikzpicture}[scale = 2.9, align=center]
\draw[->] (0,0) -- (1.1,0);
\draw[->] (0,0) -- (0,1.1);
\draw (0,0) -- (1,0) -- (0,1) -- (0,0);
\draw (0,0) node[below]{$0$};
\draw (1,0) node[below]{$\bar{s}$};


\filldraw [line width = 0, fill=black!40] plot[id=f3,domain=0.2:.74]  (\x,{(\x-\x^4*1.2)/1.5}) -- (.5,.5) -- (.16,.16) -- cycle;
\filldraw [fill=black!40] plot[id=f3,domain=0.2:.74]  ({(\x-\x^4*1.2)/1.5},\x) -- (.5,.5) -- (.16,.16) -- cycle;
\draw (0.5,0.5) -- (0,0);

\draw (0,0) -- (.65,.35);
\draw (.62,.4) node[right]{$\theta$};

\end{tikzpicture}
\caption{\footnotesize The conditional distribution $F(\theta|s)$ is monotone for a monotone non-increasing $\theta(s)$ where conditioned on $s$, the values are uniform from the set $\{\type|t_1+t_2= s, \min(t_1,t_2)/\max(t_1,t_2) \geq \theta(s)\}$.  For example, for any $\delta \leq \bar{s}/2$, setting $\theta(s) = \delta(1+s)/s$ defines the set of types to be the triangle $t_1,t_2 \in [\delta,\bar{s}-\delta]$, $t_1+t_2\leq \bar{s}$.}
\label{fig:exampleadditive}
\end{figure}

Similar to \autoref{s:uniform}, it is sufficient to prove the statement assuming $t_1\geq t_2$. As in \autoref{s:uniform} the sum-of-values projection, via the
divergence density equality (of \autoref{lem:framework}), pins down an
amortization $\amorev$ that is tight for any grand bundle pricing.  This tight amortization may fail to be a virtual value function because virtual surplus with respect to $\amorev$ is not pointwise optimized by a grand bundle pricing.  For
this reason, we directly define $\virt$ and then prove that it is a virtual value function for the grand bundle pricing mechanism by comparing the virtual surplus with respect to $\virt$ and $\amorev$.

\begin{definition}\label{def:sumextension} The {\em two-dimensional extension} $\virt$ of the amortization of the sum-of-values projection $\sumamorev(s) = s - \frac{1-\sumdist(s)}{\sumdens(s)}$ is:
\begin{align*}
\virtelement_1(\typeprelim) &= \frac{t_1}{t_1+t_2} \sumamorev(t_1+t_2) = t_1 - \frac{t_1}{t_1+t_2}\frac{1-\sumdist(t_1+t_2)}{\sumdens(t_1+t_2)}, \\
\virtelement_2(\typeprelim) &= \frac{t_2}{t_1+t_2} \sumamorev(t_1+t_2) = t_2 - \frac{t_2}{t_1+t_2}\frac{1-\sumdist(t_1+t_2)}{\sumdens(t_1+t_2)}.
\end{align*}
\end{definition}

The following lemma provides conditions on vector field $\virt$ such that bundle pricing maximizes virtual surplus pointwise with respect to  $\virt$ (proof in \aref{l:vsm-consistency-addproof}).  These conditions are satisfied for $\virt$ of \autoref{def:sumextension}, if $\sumamorev(s)$ is monotone non-decreasing.

\begin{lemma}
\label{l:vsm-consistency-add}
The allocation of a bundle pricing mechanism pointwise optimizes virtual surplus with respect to vector field $\virt$ for all costs $\cost \max(x_1,x_2)$ if and only if: $\virtelement_1(\type)$ and $\virtelement_2(\type)$ have the same sign, $\virtelement_1(\type)+\virtelement_2(\type)$ is only a function of $t_1+t_2$ and is monotone non-decreasing in $t_1+t_2$.
\end{lemma}

Given \autoref{l:vsm-consistency-add}, the remaining steps in proving that $\virt$ is a virtual value function is showing that it is  a tight amortization for grand bundle pricing.  The following lemma proves tightness (proof in \aref{l:sd-proj-addproof}).

\begin{lemma}
\label{l:sd-proj-add}
The expected revenue of a bundle pricing is equal to its expected
virtual surplus with respect to the two-dimensional extension $\virt$ of the sum-of-values projection (\autoref{def:sumextension}).
\end{lemma}


The rest of this section shows that $\virt$ provides an upper bound on revenue of any mechanism. For that, we study the existence of a tight canonical amortization $\amorev$ such that the virtual surplus of any incentive compatible mechanism with respect to $\virt$ upper bounds its virtual surplus with respect to $\amorev$ (any such $\amorev$ must be tight for any bundle pricing since $\virt$ is) and invoke  \autoref{prop:amorev-to-virt}.    Define the \emph{equi-quantile} function $\curve_q(s)$ such that conditioned on $s$, the probability that $t_2\leq  \curve_q(s)$ is equal to $q$.

\begin{lemma}\label{lem:additiveamortization}
If the conditional distribution $F(\theta|s)$ is monotone non-decreasing in $s$, then there exists a canonical amortization $\amorev(\type) = \type - \amortil(\type)/\dens(\type)$ such that $\expect{\alloc(\type)\cdot (\virt(\type)-\amorev(\type))}\geq 0$ for all incentive compatible mechanisms.  For any $\type$, $\amortil(\type)$ is tangent to the equi-quantile curve crossing $\type$. 
\end{lemma}

We show the following refinement of \autoref{prop:amorev-to-virt}, for any incentive compatible
allocation $\sagentmech$ and sum $s$,
\begin{align}
\label{exp2}
\expect{\sagentmech(\typeprelim) \cdot (\virt(\typeprelim) - \amorev(\typeprelim)) \given t_1 + t_2 = s } &\geq 0.
\end{align}
That is, we use a sweeping process in a single dimension and along lines with constant sum of values $s$ (see \autoref{sec:reverseengineering}). Consider the amortization $\amorev$ that, like $\virt$,
sets $\amorevelement_1(\type) + \amorevelement_2(\type) =
\sumamorev(t_1+t_2)$ but, unlike $\virt$, splits this
total amortized value across the two coordinates to satisfy the
divergence density equality.  Equation~\eqref{exp2} can be expressed in terms of this relative difference $\virtelement_1 - \amorevelement_1$ since $\sagentmech \cdot (\virt - \amorev) = (\sagentmechelement_1- \sagentmechelement_2)(\virtelement_1 - \amorevelement_1)$.
We will first show that to satisfy equation~\eqref{exp2} for all incentive compatible $\sagentmech$ it is sufficient for $\amorev$, relative to
$\virt$, to place less value on the favorite coordinate, i.e., $\amorevelement_1 \leq \virtelement_1$.  Notice that since $\amorevelement_1 + \amorevelement_2 = \virtelement_1 + \virtelement_2$ and $\virtelement_1\frac{t_2}{t_1} = \virtelement_2$, the condition $\amorevelement_1 \leq \virtelement_1$ is equivalent to the condition $\amorevelement_1\frac{t_2}{t_1} \leq \amorevelement_2$.

To calculate the expectation in equation~\eqref{exp2}, it will be
convenient to change to sum-ratio coordinate space.  For a function
$h$ on type space $\typespace$, define $h^{S\!R}$ to be its
transformation to sum-ratio coordinates, that is
\begin{align*}
h(t_1,t_2) = h^{S\!R}(t_1+t_2, \tfrac{t_2}{t_1}).
\end{align*}

Our derivation of sufficient conditions for the two-dimensional
extension of the sum-of-values projection to be an amortization
exploits two properties. First, by convexity of utility (\autoref{l:gradutil=alloc}), the change in allocation probabilities of an
incentive compatible mechanism, for a fixed sum $s$ as the ratio
$\theta$ increases, can not be more for coordinate one than coordinate two, that is, $\sagentmechcorelement{S\!R}_1(s,\theta)-\sagentmechcorelement{S\!R}_2(s,\theta)$ must be non-increasing in $\theta$ (\autoref{lem:addIC}). Second, if $\amorev$ shifts value from coordinate one to coordinate two relative to the vector field $\virt$, then, it also shifts {expected} value from coordinate one to coordinate two, conditioned on sum $t_1 + t_2 = s$ and ratio $t_2/t_1
\leq \theta$.   We then use integration by parts to show that the shift in expected value only
hurts the virtual surplus of $\amorev$ relative to $\virt$ and equation \eqref{exp2} is satisfied (\autoref{phi-sufficient-additive}, proof in \aref{phi-sufficient-additiveproof}).  Later in the section we will describe
sufficient conditions on the distribution to guarantee existence of $\amorev$ where this sufficient
condition that $\amorevelement_1\frac{t_2}{t_1} \leq \amorevelement_2$, is satisfied (\autoref{l:amortil-angle-additive}).

\begin{lemma}\label{lem:addIC}
The allocation of any differentiable incentive compatible mechanism satisfies
\begin{align*}
\frac{d}{d\theta} \sagentmechcor{S\!R}(s,\theta) \cdot (-1,1) \geq 0.
\end{align*}
\end{lemma}
\begin{proof}
The proof follows directly from \autoref{l:gradutil=alloc}.  In particular, convexity of the utility function implies that the dot product of any vector, here $(-1,1)$, and the change in gradient of utility $\sagentmech$ in the direction of that vector, here $\frac{d}{d\theta} \sagentmechcor{S\!R}(s,\theta)$, is positive. 
\end{proof}

\begin{lemma}\label{phi-sufficient-additive}
The two-dimensional extension of the sum-of-values projection $\virt$
is an amortization if there exists an amortization $\amorev$
with $\amorevelement_1(\type) + \amorevelement_2(\type) =
\sumamorev(t_1+t_2)$ that satisfies $\amorevelement_1(\type)\frac{t_2}{t_1}
\leq \amorevelement_2(\type)$.
\end{lemma}

To identify sufficient conditions for $\virt$ to be an amortization it now suffices to derive conditions under which there exists
a canonical amortization $\amorev$ satisfying $\amorevelement_1(\type) +
\amorevelement_2(\type) = \sumamorev(t_1+t_2)$ and the
condition of \autoref{phi-sufficient-additive}, i.e.,
$\amorevelement_1(\type)\frac{t_2}{t_1} \leq \amorevelement_2(\type)$.
Notice that $\amortilelement_1 \frac{t_2}{t_1} \geq \amortilelement_2$
implies that $\amorevelement_1 \frac{t_2}{t_1} \leq \amorevelement_2$
because
\begin{align*}
\frac{t_2}{t_1} \, \sdamorev_1(\type) 
  &= \frac{t_2}{t_1}\,\big(t_1 - \frac{\sdamortil_1(\type)}{\dens(\type)}\big) 
  = \frac{t_2}{t_1}\,\big(t_1 - \frac{\sdamortil_1(\type)}{\dens(\type)}\big)  
  \leq t_2 - \frac{\sdamortil_2(\type)}{\dens(\type)} 
  = \sdamorev_2(\type).
\end{align*}
Thus, it suffices to identify conditions under which $\amortilelement_1
\frac{t_2}{t_1} \geq \amortilelement_2$. 

The following constructs the canonical amortization $\amorev$ and specifies the direction of $\amortil$.  Similar to \autoref{s:uniform}, $\amortil$ is tangent to the equi-quantile curve, that in the section are defined by conditioning on the value for bundle $s$.  The proof is similar to the proof of \autoref{lem:strongamort} and is deferred to \aref{l:amortil-angle-additiveproof}. 

\begin{lemma}\label{l:amortil-angle-additive}\label{lem:gamma}
A canonical amortization $\amorev = \type - \amortil/\dens$ satisfying $\amorevelement_1(\type) + \amorevelement_2(\type) = \sumamorev(t_1+t_2)$ exists and is unique, where $\amortil(\type)$ is tangent to the equi-quantile curve crossing $\type$.
\end{lemma}

\begin{proof}[Proof of \autoref{lem:additiveamortization}]
The assumption that $F(\theta|s)$ is monotone implies that the equi-quantile curves are \thetamonotone.  The tangency property of \autoref{lem:gamma} implies that $\amortilelement_1
\frac{t_2}{t_1} \geq \amortilelement_2$ and subsequently $\amorevelement_1(\type)\frac{t_2}{t_1}
\leq \amorevelement_2(\type)$.  \autoref{phi-sufficient-additive} then implies that $\virt$ is an amortization.
\end{proof}

\begin{proof}[Proof of \autoref{thm:additiveif}] 
\autoref{lem:additiveamortization} showed that $\virt$ is an amortization.  \autoref{l:vsm-consistency-add} showed that the allocation of bundle pricing maximizes virtual surplus with respect to $\virt$, and \autoref{l:sd-proj-add} showed that $\virt$ is tight for bundle pricing.  Invoking \autoref{prop:vvf=>IC+opt} completes the proof.
\end{proof}
\section{Discussion}\label{s:conclusions}
We briefly discuss the generality of the design of virtual values can be applied to prove optimality of mechanisms.  In the context of the simple favorite outcome mechanism studied in this paper, the method gives very general and nearly tight conditions of optimality.  However, the approach has certain limitations.  For example, with linear values and costs, pointwise optimization of surplus can result only in deterministic outcomes, whereas randomized outcomes are know to be optimal in various settings.\footnote{Randomized outcomes might arise from surplus optimization if virtual values are equal.  However, imposing such a constraint will severely limit the applicability of the approach.}  In spite of that, virtual surplus optimization can create internal allocations with nonlinear valuations and costs, as studied for example by \citet{A96} and \citet{RC98}.

\bibliographystyle{apalike}
\bibliography{bibs}

\newpage
\appendix
\section{Proofs from Section 4}
\label{a:missing}
This section includes proofs from  \autoref{s:uniform}.

\subsection{Proof of \autoref{t:uniform-pricing-strong-amortization-many-agents-items}}\label{a:proofmostgeneral}
\begin{reptheorem}{t:uniform-pricing-strong-amortization-many-agents-items}
A favorite-outcome projection mechanism is optimal for an item with $m\geq1$ configurations, multiple independent agents, and any service cost $\cost \geq 0$, if the distribution of each agent is \quasisymmetric and (a) the favorite outcome projection has monotone
non-decreasing amortization $\phi_{\max}(v) = v -
\frac{1-\sdpdist(v)}{\sdpdens(v)}$ and (b) $\dist(\theta_j|\val,i, \boldsymbol{q}_{<j})$ is monotone non-increasing in $\val$ for all $i$, $j$, $\theta_j$, and $\boldsymbol{q}_{<j}$.
\end{reptheorem}

\begin{proof}
The construction extends the construction of \autoref{t:uniform-pricing-strong-amortization-general}.  Let outcome 1 be the favorite outcome.  For $\boldsymbol{q}$, let $\curve^{\boldsymbol{q}}(t_1)$ be a function that maps $t_1$ to $(t_2,\ldots,t_m)$ such that $\boldsymbol{q}(\type) = \boldsymbol{q}$.  Define  $\amortil$ by integrating by parts along the curves $\curve^{\boldsymbol{q}}(t_1)$.  This defines $\amortilelement_1(\type) =  \frac{1- \sdpdist(t_1)}{\sdpdens(t_1)} \dens(\type)$, and $\amortilelement_i(\type) = \amortilelement_1(\type)\partial_{t_1} \curve_i^{\boldsymbol{q}}(t_1)$.  The assumptions of the theorem also implies that  $\amortilelement_i(\type) - (t_i/t_1) \amortilelement_1(\type)\leq 0$.  As a result, $\amorevelement_i(\type) \leq (t_i/t_1) \amorevelement_1(\type)$.

With multiple agents, $m\geq 1$, and uniform service cost $c$, ex-post optimization of virtual surplus allocates the agent with the highest positive virtual value.  The argument above shows that the highest positive virtual value of any agent corresponds to the favorite outcome of that agent, and is equal to the virtual value of the single-dimensional projection.
\end{proof}

\subsection{Product Distributions Over Values}\label{app:independent}
In this section we derive conditions that prove optimality of the single-dimensional projection for product distributions over values.
\begin{theorem}
Uniform pricing is optimal for any cost $c$ for an instance with two outcomes where the value for each outcome is drawn independently from a distribution with density proportional to $e^{h(\log(x))}$.
\end{theorem}
We will show that the distribution satisfies the conditions of \autoref{t:uniform-pricing-strong-amortization-many-agents-items}.  In order to show that $\dist(\theta|\val)$ is monotone in $\val$, we show that the joint distribution of $\theta$ and $\val$ satisfies the stronger property of affiliation. That is,
\begin{align*}
\denscor{M\!R}(t_1,\theta) \times \denscor{M\!R}(t'_1,\theta') \geq \denscor{M\!R}(t_1,\theta') \times \denscor{M\!R}(t'_1,\theta), \qquad &\forall t_1 \leq t'_1, \theta \leq \theta',
\end{align*}

\noindent where $\denscor{M\!R}(t_1,\theta) = f(t_1,t_1\theta)$ is the joint distribution of $t_1$ and $\val$. Since the distribution is a product one, this implies that $\denscor{M\!R}(t_1,\theta) = f_1(t_1) f_2(t_1\theta)$. Notice that pair of values $t\theta'$ and $t'\theta$ have the same geometric mean as the pair $t\theta$, $t'\theta'$. Also given the assumptions, $t\theta \leq t'\theta, t\theta' \leq t'\theta'$. Since $\dens(x) = \eta \cdot e^{h(\log(x))}$,
\begin{eqnarray*}
f_2(t_1\theta) \times f_2(t'_1\theta') \geq f_2(t\theta') \times f_2(t'\theta).
\end{eqnarray*}

Multiplying both sides by $f_1(t_1) \times f_1(t'_1)$ we get
\begin{eqnarray*}
f_1(t_1) f_2(t_1\theta) \times f_1(t'_1) f_2(t'_1\theta') \geq f_1(t_1) f_2(t_1\theta') \times f_1(t'_1) f_2(t'_1\theta),
\end{eqnarray*}

\noindent which since the distribution is a product distribution implies that
\begin{align*}
\denscor{M\!R}(t_1,\theta) \times \denscor{M\!R}(t'_1,\theta') \geq \denscor{M\!R}(t_1,\theta') \times \denscor{M\!R}(t'_1,\theta).
\end{align*}

To complete the proof, we need to show that $\sdpdist$ is regular.  This is the case because $\sdpdens(\val) = \dist(\val)\dens(\val)$, $\dens(\val) = \eta \cdot e^{h(\log(\val))}$ is monotone in $\val$ by monotonicity of $h$.
\subsection{Proof of \autoref{thm:unitdemandironing}}\label{a:ironing}

\begin{reptheorem}{thm:unitdemandironing}
A favorite-outcome projection mechanism is optimal for an item with $m=2$ configurations, multiple independent agents, and any service cost $\cost \geq 0$, if the distribution of each agent is \quasisymmetric with convex equi-quantile curves.
\end{reptheorem}

We will design a virtual value function $\virt$ from the canonical amortization $\amorev$ satisfying conditions of \autoref{lem:framework}. Importantly, $\virt$ satisfies the monotonicity of $\virtelement_1$ without requiring regularity of the distribution of the favorite item projection. We will start by defining a mapping between the type space and a two-dimensional {quantile} space. We will then use Myerson's ironing to pin down the first coordinate $\virtelement_1$ of the amortization. The second component $\virtelement_2$ is then defined such that the expected virtual surplus with respect $\virt$ upper bounds revenue for all incentive compatible mechanisms. To do this, we invoke integration by parts along curves defined by the quantile mapping, and then use incentive compatibility to identify a direction that the vector $\virt - \amorev$ may have for $\virt$ to be an upper bound on revenue. We use this identity to solve for $\virtelement_2$, and finally identify conditions such that optimization of $\virt$ gives uniform pricing.

We first transform the value space to quantile space using following mappings. Recall from \autoref{s:uniform} that $\sdpdist$ and $\sdpdens$ are the distribution and the density functions of the favorite item projection. Define the first quantile mapping
\begin{align*}
q_1(t_1,t_2) &= 1- \sdpdist(t_1)
\intertext{to be the probability that a random draw $t'_1$ from $\sdpdist$ satisfies $t'_1\geq t_1$, and the second quantile mapping}
q_2(t_1,t_2) &= 1- \frac{\int_{t'_2=0}^{t_2}\dens(t_1,t'_2)\; \dd t'_2}{\sdpdens(t_1)}
\intertext{where $\sdpdens(t_1) = \int_{0}^{t_1} \dens(t_1,t'_2) \; \dd t'_2$ to the probability that a random draw $\type'$ from a distribution with density $f$, conditioned on $t'_1=t_1$, satisfies $t'_2\geq t_2$. The determinant of the Jacobian matrix of the transformation is}
\begin{vmatrix}
\frac{\partial q_1}{\partial t_1}&\frac{\partial q_1}{\partial t_2}\\
\frac{\partial q_2}{\partial t_1}&\frac{\partial q_2}{\partial t_1}
\end{vmatrix}
&=
\begin{vmatrix}
-\sdpdens(t_1)&0\\
\frac{\partial q_2}{\partial t_1}&-\frac{\dens(t_1,t_2)}{\sdpdens(t_1)}
\end{vmatrix}
= \dens(t_1,t_2).
\intertext{As a result, we can express revenue in quantile space as follows}
\int \int \sagentmech(\typeprelim) \cdot \amorev(\typeprelim) \; f(\typeprelim)  \; \dd \typeprelim &= \int_{q_1=0}^1 \int_{q_2=0}^1 \sagentmech^Q(\quantile) \cdot \amorev^Q(\quantile) \; \dd \quantile,
\end{align*}
\noindent where $\sagentmech^Q$ and $\amorev^Q$ are representations of $\sagentmech$ and $\amorev$ in quantile space. In particular, $\amorevelement_1^Q(\quantile) = \sdpamorev(t_1(q_1))$ might not be monotone in $q_1$. In what follows we design the amortization $\virt^Q$ using $\amorev^Q$.

We now derive $\virt^Q$ from the properties it must satisfy. In particular, we require $\virtelement^Q_1(\quantile) = \virtelement^Q_1(q_1)$ to be a monotone non-decreasing function of $q_1$, and that $\virtelement^Q_1(\quantile) \geq \virtelement^Q_2(\quantile)$ whenever either is positive. These properties will imply that a point-wise optimization of $\virt^Q$ will result in an incentive compatible allocation of only the favorite item, such that $\sagentmechelement_1^Q(\quantile) = \sagentmechelement_1^Q(q_1)$, and $\sagentmechelement_2^Q(\quantile) =0$ (which is the case for the allocation of uniform pricing). Note that for any such allocation,
\begin{align*}
\int_{q_1=0}^1 \int_{q_2=0}^1 \sagentmech^Q(\quantile) \cdot \amorev^Q(\quantile) \; \dd \quantile &= \int_{q_1} \sagentmechelement^Q_1(q_1) \amorevelement^Q_1(q_1) \; \dd q_1.
\intertext{Similarly, for any such allocation,}
\int_{q_1=0}^1 \int_{q_2=0}^1 \sagentmech^Q(\quantile) \cdot \virt^Q(\quantile) \; \dd \quantile &= \int_{q_1} \sagentmechelement^Q_1(q_1) \virtelement^Q_1(q_1) \; \dd q_1.
\end{align*}
\noindent We can therefore use Myerson's ironing and define $\virtelement^Q_1$ to be the derivative of the convex hull of the integral of $\amorevelement^Q_1$. This will imply that $\virt^Q$ upper bounds revenue for any allocation that satisfies $\sagentmechelement^Q_1(\quantile) = \sagentmechelement^Q_1(q_1)$, and $\sagentmechelement^Q_2(\quantile) = 0$, with equality for the allocation that optimizes $\virt^Q$ pointwise.

We will next define $\virtelement^Q_2$ such that $\virt^Q$ upper bounds revenue for \emph{all} incentive compatible allocations. That is, we require that for all incentive compatible $\sagentmech$,
\begin{align*}
\int \int \sagentmech^Q(\quantile) \cdot (\virt^Q - \amorev^Q)(\quantile) \; \dd \quantile &\geq 0.
\end{align*}
\noindent Using integration by parts we can write
\begin{align*}
\hspace{-20 mm}\int \int \sagentmech^Q(\quantile) \cdot (\virt^Q - \amorev^Q)(\quantile) \; \dd \quantile &= \int_{q_2} \int_{q_1} \frac{\dd}{\dd q_1} \sagentmech^Q(\quantile) \cdot \int_{q'_1\geq q_1} (\virt^Q - \amorev^Q)(q'_1,q_2) \; \dd q'_1 \; \dd q_1 \; \dd q_2.
\end{align*}
\noindent Incentive compatibility implies that the dot product of any vector and the change in allocation rule in the direction of that vector is non-negative (\autoref{l:gradutil=alloc}). In particular this must be true for the tangent vector to equi-quantile curve parameterized by $q_2$. Thus incentive compatibility of $\sagentmech$ implies that the above expression is positive if the vector that is multiplied by $\frac{\dd}{\dd q_1} \sagentmech^Q(\quantile)$ is tangent to the equi-quantile curve $(t_1(q'_1,q_2),t_2(q'_1,q_2)), 0\leq q'_1 \leq q_1$ at $q'_1=q_1$,
\begin{align*}
\frac{\int_{q'_1\geq q_1} (\virtelement^Q_2 - \amorevelement_2^Q)(q'_1,q_2) \; \dd q'_1}{\int_{q'_1\geq q_1} (\virtelement^Q_1 - \amorevelement_1^Q)(q'_1,q_2) \; \dd q'_1} &= \frac{\frac{\dd}{\dd q_1}t_2(\quantile)}{\frac{\dd}{\dd q_1}t_1(\quantile)}.
\end{align*}

\noindent We will set $\virtelement^Q_2$ to satisfy the above equality. In particular, define for simplicity $\mu(\quantile) = \frac{\frac{\dd}{\dd q_1}t_2(\quantile)}{\frac{\dd}{\dd q_1}t_1(\quantile)}$ and take derivative of the above equality with respect to $q_1$
\begin{align*}
\hspace{-10 mm}\virtelement^Q_2(\quantile) &= \amorevelement_2^Q(\quantile) + (\virtelement^Q_1 - \amorevelement_1^Q)(\quantile)\cdot\mu(\quantile) - \int_{q'_1\geq q_1} (\virtelement^Q_1 - \amorevelement_1^Q)(q'_1,q_2) \; \dd q'_1 \cdot \frac{\dd}{\dd q_1}\mu(\quantile).
\end{align*}

As a result, $\virt^Q$ defined above is a tight amortization if its optimization indeed gives uniform pricing. The next lemma formally states the above discussion.

\begin{lemma}\label{virtisweakamort}
The virtual surplus, with respect to $\virt^Q$ of any incentive compatible allocation $\sagentmech$ upper bounds its revenue. If $\sagentmechelement_1$ is only a function of $q_1$ (equivalently $t_1$), $\sagentmechelement'_1(q_1)=0$ whenever $\int_{q'_1\geq q_1}(\virtelement^Q_1-\amorevelement^Q_1)(q'_1)\; \dd q'_1>0$,  and $\sagentmechelement_2(\quantile)=0$ for all $\quantile$, the expected virtual surplus with respect to $\virt^Q$ equals revenue.
\end{lemma}

We will finally need to verify that $\virt^Q$ also satisfies the properties required for ex-post optimization. \autoref{ironingcomparevirt} below identifies convexity of equi-quantile curves as a sufficient condition.  The proof requires the following technical lemma.

\begin{lemma}\label{lem:tmorethanphi}
The amortization $\virtelement$ satisfies $\virtelement_1(\type) \leq t_1$.
\end{lemma}
\begin{proof}
In \emph{un-ironed} regions, that is whenever $\virtelement_1 = \amorevelement_1$, by definition we have $\virtelement_1(\type) = t_1 - \frac{1-\sdpdist(t_1)}{\sdpdens(t_1)} \leq t_1$. If the curve is ironed between $q_1$ and $q'_1 \geq q_1$, then $\virtelement^Q_1$ is the derivative of convex hull of $\amorevelement_1^Q$, which is $\int_0^q t_1(q') - \frac{q}{\sdpdens(t_1(q))} \; \dd q' = qt_1(q)$. Thus, for all $q''_1$ with $q_1 \leq q''_1 \leq q'_1$ we have
\begin{align*}
\virtelement^Q_1(q''_1) &= \frac{q'_1t'_1(q'_1)-q_1t_1(q_1)}{q'_1-q_1}\\
&\leq \frac{q'_1t_1(q'_1) - q_1t_1(q'_1)}{q'_1 - q_1}\\
&=t_1(q'_1) \leq t_1(q''_1).
\end{align*}
\end{proof}
\begin{lemma}\label{ironingcomparevirt}
If the equi-quantile curves are convex for all $q_2$, the amortization $\virt^Q$ defined above satisfies $\theta(\quantile) \virtelement^Q_1(\quantile) \geq \virtelement^Q_2(\quantile)$. As a result, $\virtelement^Q_1\geq \virtelement^Q_2$ whenever either is positive.
\end{lemma}
\begin{proof}
\autoref{lem:strongamort} showed that $\amortil$ is tangent to the equi-quantile curves.  This implies that $\amorevelement_1^Q(\quantile)\mu(\quantile) - \amorevelement_2^Q(\quantile) = t_1(\quantile)\mu(\quantile) - t_2(\quantile)$. By rearranging the definition of $\amorevelement_2$ we get
\begin{align*}
\hspace{-6.8ex}\virtelement^Q_1(\quantile)\mu(\quantile) - \virtelement^Q_2(\quantile) &= \amorevelement_1^Q(\quantile)\mu(\quantile) - \amorevelement_2(\quantile) +  \int_{q'_1\geq q_1} (\virtelement^Q_1 - \amorevelement_1^Q)(q'_1,q_2) \; \dd q'_1 \cdot \frac{\dd}{\dd q_1}\mu(\quantile)\\
&= t_1(\quantile)\mu(\quantile) - t_2(\quantile) + \int_{q'_1\geq q_1} (\virtelement^Q_1 - \amorevelement_1^Q)(q'_1,q_2) \; \dd q'_1 \cdot \frac{\dd}{\dd q_1}\mu(\quantile)\\
&\geq t_1(\quantile)\mu(\quantile) - t_2(\quantile),
\intertext{where the inequality followed since by definition of $\virtelement^Q_1$, we have $\int_{q'_1\geq q_1} (\virtelement^Q_1 - \amorevelement_1^Q)(q'_1,q_2) \; \dd q'_1 \geq 0$, and $\frac{\dd}{\dd q_1}\mu(\quantile)\geq 0$ by the assumption of the lemma. We can now rearrange the above inequality and write}
t_2(\quantile) - \virtelement^Q_2(\quantile) &\geq \mu(\quantile) (t_1(\quantile) - \virtelement^Q_1(\quantile))\\
&\geq \theta(\quantile) (t_1(\quantile) - \virtelement^Q_1(\quantile)),
\end{align*}
\noindent where the inequality followed since convexity of equi-quantile curves imply that $\mu(\quantile) \geq \theta(\quantile)$, and by \autoref{lem:tmorethanphi}, $t_1(\quantile) - \virtelement^Q_1(\quantile)\geq 0$.

We can now use the above inequality to write
\begin{align*}
\theta(\quantile) \virtelement^Q_1(\quantile) &= \theta(\quantile) (t_1(\quantile) + (\virtelement^Q_1(\quantile)-t_1(\quantile))\\ &= t_2(\quantile) + \theta(\quantile) (\virtelement^Q_1(\quantile)-t_1(\quantile)) \\
&\geq t_2(\quantile) + \virtelement^Q_2(\quantile)-t_2(\quantile)\\
&= \virtelement^Q_2(\quantile).
\end{align*}
\end{proof}

\begin{proof}[Proof of \autoref{thm:unitdemandironing}]
Combining \autoref{virtisweakamort} and \autoref{ironingcomparevirt} proves the theorem.
\end{proof}

\section{Proofs from Section 5}
This section contains proofs from \autoref{a:add bundle}.

\subsection{Proof of \autoref{l:vsm-consistency-add}}\label{l:vsm-consistency-addproof}
\begin{replemma}{l:vsm-consistency-add}
The allocation of a bundle pricing mechanism pointwise optimizes virtual surplus with respect to vector field $\virt$ for all costs $\cost \max(x_1,x_2)$ if and only if: $\virtelement_1(\type)$ and $\virtelement_2(\type)$ have the same sign, $\virtelement_1(\type)+\virtelement_2(\type)$ is only a function of $t_1+t_2$ and is monotone non-decreasing in $t_1+t_2$.
\end{replemma}
\begin{proof}
We need to show that for the uniform price $p$, the allocation
function $\sagentmech$ of posting a price $p$ for the bundle
optimizes $\amorev$ pointwise. Pointwise optimization of $\sagentmech
\cdot \virt$ will result in $\sagentmech = (1,1)$ whenever
$\virtelement_1+\virtelement_2 \geq c$, and $\sagentmech = (0,0)$ otherwise.
\end{proof}

\subsection{Proof of \autoref{l:sd-proj-add}}\label{l:sd-proj-addproof}
\begin{replemma}{l:sd-proj-add}
The expected revenue of a bundle pricing is equal to its expected
virtual surplus with respect to the two-dimensional extension $\virt$ of the sum-of-values projection (\autoref{def:sumextension}).
\end{replemma}

\begin{proof}
Let $\sagentmech^p$ be the allocation corresponding to posting
price $p$ for the bundle, that is $\sagentmechelement_1^p(\type) =
\sagentmechelement_2^p(\type) = 1$ if $t_1 +t_2 \geq p$, and $\sagentmechelement_1^p(\type) = \sagentmechelement_2^p(\type)= 0$
otherwise. We will show that the virtual surplus of $\sagentmech^p$ is equal to the revenue of posting price
$p$, $R(p) = p (1-\dist_{sum}(p))$. The virtual surplus is
\begin{align*}
\int_{\typeprelim\in T} (\sagentmech^p \cdot \amorev f)(\typeprelim) \; \dd \typeprelim 
&= \int_{\type \in T} \sagentmech^p(t_1,t_2) \cdot \amorev(t_1,t_2) f(t_1,t_2) \; \dd \type\\
& = \int_{\type \in T, t_1 + t_2\geq p} \phi_{sum}(t_1+t_2) f(t_1,t_2) \; \dd \type.\\
&= - \int_{s \geq p} \frac{\dd}{\dd s} (s(1-\dist_{sum}(s)) \; \dd s \\
&= R(p) - R(1) = R(p).
\end{align*}
\end{proof}

\subsection{Proof of \autoref{phi-sufficient-additive}}\label{phi-sufficient-additiveproof}
\begin{replemma}{phi-sufficient-additive}
For any symmetric distribution over values for items, the two-dimensional extension of the sum-of-values projection $\virt$
is an amortization of revenue if there exists an amortization of revenue $\amorev$
with $\amorevelement_1(\type) + \amorevelement_2(\type) =
\sumamorev(t_1+t_2)$ that satisfies $\amorevelement_1(\type)\frac{t_2}{t_1}
\leq \amorevelement_2(\type)$.
\end{replemma}
\begin{proof}
Without loss of generality, in proving equation~\eqref{exp2} we can assume that the allocation is symmetric.  This is because by symmetry of the distribution, there exists an optimal mechanism that is also symmetric. Therefore, it is sufficient to prove the lemma only for symmetric incentive compatible allocations (in particular, we assume that $x_1(t_1,t_1) = x_2(t_1,t_1)$ for all $t_1$).\footnote{In general, when optimal mechanisms are known to satisfy a certain property, the inequality of amortization needs to be shown only for mechanisms satisfying that property.}

Fix the sum $s = t_1+t_1$.  Denote the expected difference between
$\virt$ and $\amorev$ conditioned on $t_2/t_1 \leq \theta$ by:
\begin{align*}
\cumdiff(s,\theta) &=  \int_{\theta'=0}^\theta [\virt - \amorev]^{S\!R}(s,\theta') \denscor{S\!R}(s,\theta')\; \frac{s}{1+\theta} \dd \theta'.
\end{align*}
We will only be interested in three properties of $\cumdiff$:
\begin{enumerate}[(a)]
\item \label{p:diff} $\cumdiffelement_2(s,\theta) = - \cumdiffelement_1(s,\theta)$, i.e., this is the expected amount of value shifted from coordinate one to coordinate two of $\virt$ relative to $\amorev$. This follows from the fact that $\amorevelement_1(\type) + \amorevelement_2(\type) = \virtelement_1(\type) + \virtelement_2(\type) = \sumamorev(t_1+t_2)$.
\item \label{p:shift} $\cumdiffelement_2(s,\theta) \geq 0$, i.e., this shift is non-negative according to the assumption of the lemma.  
\item \label{p:int} $\cumdiff(s,0) = \bfzero$, as the range of the integral is empty at $\theta=0$.
\end{enumerate}
Write the left-hand side of equation~\eqref{exp2} as:
\begin{align*}
&\hspace{-6.8ex} \expect{\sagentmech(\typeprelim) \cdot (\virt(\typeprelim) - \amorev(\typeprelim)) \given t_1 + t_2 = s }\\
&=\int_{\theta=0}^1 \sagentmechcor{S\!R}(s,\theta) \cdot [\virt - \amorev]^{S\!R}(s,\theta) \denscor{S\!R}(s,\theta) \frac{s}{1+\theta} \; \dd \theta\\
&= \int_{\theta=0}^1 \sagentmechcor{S\!R}(s,\theta) \cdot \frac{d}{d\theta} \int_{\theta'=0}^\theta [\virt - \amorev]^{S\!R}(s,\theta') \denscor{S\!R}(s,\theta') \frac{s}{1+\theta'} \; \dd \theta' \; \dd \theta. \\
\intertext{Substituting $\cumdiff$ into the integral above, we have}
&=  \int_{\theta=0}^1 \sagentmechcor{S\!R}(s,\theta) \cdot \frac{d}{d\theta} \cumdiff(s,\theta) \; \dd \theta\\
&=  \sagentmechcor{S\!R}(s,\theta) \cdot \cumdiff(s,\theta)\big|_{\theta=0}^1 - \int_{\theta=0}^1 \frac{d}{d\theta} \sagentmechcor{S\!R}(s,\theta) \cdot  \cumdiff(s,\theta) \; \dd \theta.\\
&= - \int_{\theta=0}^1 \frac{d}{d\theta} \sagentmechcor{S\!R}(s,\theta) \cdot  \cumdiff(s,\theta) \; \dd \theta \; \dd s \\
&\geq 0.
\end{align*}
The second equality is integration by parts.  The third equality follows because the first term on the left-hand side is zero:
For $\theta = 0$,
$\cumdiff(s,\theta)=\bfzero$ by property~\eqref{p:int}; for $\theta=1$, $\sagentmechcorelement{S\!R}_1(s,\theta) =
\sagentmechcorelement{S\!R}_2(s,\theta)$ by symmetry, and
$\cumdiffelement_1(s,\theta) = -\cumdiffelement_2(s,\theta)$ by property~\eqref{p:diff}.  The final inequality follows from $ - \frac{d}{d\theta}
\sagentmechcor{S\!R}(s,\theta) \cdot (1,-1) \geq 0$ (\autoref{lem:addIC}) and properties~\eqref{p:diff} and~\eqref{p:shift}.
\end{proof}

\subsection{Proof of \autoref{l:amortil-angle-additive}}\label{l:amortil-angle-additiveproof}
\begin{replemma}{l:amortil-angle-additive}
A canonical amortization $\amorev = \type - \amortil/\dens$ satisfying $\amorevelement_1(\type) + \amorevelement_2(\type) = \sumamorev(t_1+t_2)$ exists, is unique, where $\amortil(\type)$ is tangent to the equi-quantile curve crossing $\type$.
\end{replemma}
\begin{proof}
We assume that $\amorev$ satisfying the requirements of the lemma exists, derive the closed form suggested in the lemma, and then verify that the derived $\amorev$ indeed satisfies all the required properties.  We fix
$s$ and $q$ and apply the divergence theorem to $\amortil$ on
the subspace of type space to the right of $t_1+t_2=s$ and below $C_q$. More formally, divergence theorem is applied to the set of types $T(s,q) = \{\type'\in \typespace|t'_1 +t'_2\geq s; \dist(t_2|s) \leq q\}$. The divergence theorem equates the integral of the
orthogonal magnitude of vector field $\amortil$ on the boundary of the subspace to the
integral of its divergence within the subspace.  As the upper boundary
of this subspace is $C_q$, one term in this
equality is the integral of $\amortil(\type')$ with the upward
orthogonal vector to $C_q$ at $\type'$.  Differentiating this integral with respect to $t_1$ gives the desired quantity.
\begin{align}
&\hspace{-6.8ex}\int_{\type' \in \text{TOP}(s,q)}\normal(\type') \cdot \amortil(\type') \, \dd \type' \nonumber\\
&= \int_{\type' \in T(s,q)} \nabla \cdot \amortil(\type') \, \dd \type'
-\int_{\type' \in \{\text{RIGHT,BOTTOM,LEFT}\}(s,q)}  \normal(\type') \cdot \amortil(\type') \, \dd \type'.\label{eq:div2}
\intertext{Using divergence density equality and boundary orthogonality the right hand side becomes}
&= -\int_{\type' \in T(s,q)} \dens(\type') \, \dd \type'
-\int_{\type' \in \{\text{LEFT}\}(s,q)}  \normal(\type') \cdot \amortil(\type') \, \dd \type' \nonumber\\
&= -q(1-\sumdist(s))
-\int_{\type' \in \{\text{LEFT}\}(\quantile)}  \normal(\type') \cdot \amortil(\type') \, \dd \type' \nonumber
\end{align}
\noindent where the last equality followed directly from definition of $T(s,q)$. By definition of $\amortil$, and since normal $\normal$ at the left boundary is $(-1,-1)$,
\begin{align*}
\int_{\type' \in \{\text{LEFT}\}(s,q)}  \normal(\type') \cdot \amortil(\type') \, \dd \type'&=-\frac{1-\sumdist(s)}{\sumdens(s)}\int_{t'_2\leq C_q(t_1)} \dens(t_1,t'_2)\; \dd t'_2\\
&= -\frac{1-\sumdist(s)}{\sumdens(s)} q \sumdens(s) \\
&= - (1-\sumdist(s)) q
\end{align*}
As a result, the right hand side of equation~\eqref{eq:div2} sums to zero, and we have
\begin{align*}
\int_{\type' \in \text{TOP}(s,q)}\normal(\type') \cdot \amortil(\type') \, \dd \type' = 0.
\end{align*}
\noindent Since the above equation must hold for all $s$ and $q$, we conclude that $\amortil$ is tangent to the equi-quantile curve at any type.

\end{proof}

\end{document}